\documentclass[12pt,reqno]{amsart}
\usepackage{a4wide}

\usepackage[english]{babel}

\usepackage{amsmath,amssymb,amsthm,graphicx,amsfonts,url,color,enumerate,dsfont,stmaryrd,mathabx}

\usepackage{color} \definecolor{bleu_sombre}{rgb}{0,0,0.6}  \definecolor{rouge_sombre}{rgb}{0.8,0,0}\definecolor{vert_sombre}{rgb}{0,0.6,0}
\usepackage[plainpages=false,colorlinks,linkcolor=bleu_sombre,citecolor=rouge_sombre,urlcolor=vert_sombre,breaklinks]{hyperref}

\theoremstyle{plain}
\newtheorem{theorem}{Theorem}[section] 
\newtheorem*{theorem*}{Theorem}
\newtheorem{proposition}[theorem]{Proposition}
\newtheorem*{proposition*}{Proposition}

\newtheorem*{corollary*}{Corollary}
\newtheorem{lemma}[theorem]{Lemma}

\theoremstyle{definition}

\newtheorem*{definition*}{Definition}

\theoremstyle{remark}
\newtheorem{remark}[theorem]{Remark}

\newtheorem*{example*}{Example}

\makeatletter

\@addtoreset{equation}{section}  
\makeatother

\newcommand {\limt}[2]{\xrightarrow[#1 \to #2]{}}
\newcommand{\fonc}[4] { \left\{ \begin{array}{ccc} #1 & \to & #2 \\ #3 & \mapsto & #4 \end{array} \right. }

\newcommand{\abs}[1]{\left\vert #1\right\vert}        
\newcommand{\nr}[1]{\left\Vert #1\right\Vert}         
\newcommand{\innp}[2]{\left< #1 , #2 \right>}         

\newcommand{\set}[1]{\left\{ #1 \right\}}		

\renewcommand{\leq}{\leqslant}	\renewcommand{\geq}{\geqslant}

\newcommand{\inv}{^{-1}}

\newcommand{\st}{\,:\,}
         
\renewcommand{\Re}{\mathsf{Re}}        
\renewcommand{\Im}{\mathsf{Im}}

\newcommand{\seq}[2]{\left({#1}_{#2}\right)_{#2 \in\N}} 

\newcommand{\Dom}{\mathsf{Dom}}
\newcommand{\Sp}{\mathsf{Sp}}
\renewcommand{\ker}{\mathsf{ker}}

\newcommand{\R}{\mathbb{R}}	\newcommand{\Q}{\mathbb{Q}}	\newcommand{\C}{\mathbb{C}}
\newcommand{\N}{\mathbb{N}}	\newcommand{\Z}{\mathbb{Z}}	
    \newcommand{\T}{\mathbb{T}}

\renewcommand{\a}{\alpha}\renewcommand{\b}{\beta}\newcommand{\g}{\gamma}\newcommand{\G}{\Gamma}\renewcommand{\d}{\delta}\newcommand{\e}{\varepsilon}\newcommand{\z}{\zeta} \renewcommand{\th}{\theta}\renewcommand{\k}{\kappa}\renewcommand{\l}{\lambda}\renewcommand{\L}{\Lambda}\newcommand{\m}{\mu}\newcommand{\f}{\varphi}\newcommand{\h}{\chi}\renewcommand{\O}{\Omega}

\newcommand{\Hc}{{\mathcal H}}\newcommand{\Jc}{{\mathcal J}}\newcommand{\Nc}{{\mathcal N}}\newcommand{\Oc}{{\mathcal O}}\newcommand{\Tc}{{\mathcal T}}\newcommand{\Vc}{{\mathcal V}}\newcommand{\Zc}{{\mathcal Z}}

\newcommand{\stepp}{\noindent {\bf $\bullet$}\quad }

\usepackage{mathrsfs}

\DeclareMathOperator{\cotan}{cotan}

\newcommand{\nwc}{\newcommand}
\nwc{\vareps}{\varepsilon}
\nwc{\Oph}{\operatorname{Op}_\hbar}
\nwc{\la}{\langle}
\nwc{\ra}{\rangle}

\nwc{\mf}{\mathbf} 
\nwc{\blds}{\boldsymbol} 
\nwc{\ml}{\mathcal} 

\nwc{\IN}{\mathbb{N}}
\nwc{\IR}{\mathbb{R}}
\nwc{\IZ}{\mathbb{Z}}
\nwc{\IC}{\mathbb{C}}
\nwc{\IT}{\mathbb{T}}
\nwc{\IS}{\mathbb{S}}

\newcommand{\BG}{{\mathsf{BG}}}

\newcommand{\Card}[1]{\# #1}

\newcommand{\NN}{\llbracket 1, N\rrbracket}

\DeclareMathOperator{\sh}{sh}
\renewcommand{\sinh}{\sh}

\begin{document}

\title{Spectrum of a non-selfadjoint quantum star graph}

\author[Gabriel Rivi\`ere]{Gabriel Rivi\`ere}

\address[G. Rivi\`ere]{Laboratoire de math\'ematiques Jean Leray (U.M.R. CNRS 6629), Universit\'e de Nantes, 2 rue de la Houssini\`ere, BP92208, 44322 Nantes C\'edex 3, France.}
\email{gabriel.riviere@univ-nantes.fr}

\author[Julien Royer]{Julien Royer}

\address[J. Royer]{Institut de Math\'ematiques de Toulouse (U.M.R. CNRS 5219), Universit\'e Toulouse 3, 118 route de Narbonne, 31062 Toulouse C\'edex 9,  France.}
\email{julien.royer@math.univ-toulouse.fr }

\begin{abstract} 
We study the spectrum of a quantum star graph with a non-selfadjoint Robin condition at the central vertex. We first prove that, in the high frequency limit, the spectrum of the Robin Laplacian is close to the usual spectrum corresponding to the Kirchhoff condition. Then, we describe more precisely the asymptotics of the difference in terms of the Barra-Gaspard measure of the graph. This measure depends on the arithmetic properties of the lengths of the edges. As a by-product, this analysis provides a Weyl Law for non-selfadjoint quantum star graphs and it gives the asymptotic behaviour of the imaginary parts of the eigenvalues. 
\end{abstract}

\maketitle

\section{Introduction}

\subsection{Quantum star graphs}

We consider a compact metric star graph $\Gamma$. It is defined, for some integer $N \geq 2$, by a set of $N$ edges $e_1,\dots,e_N$ of respective lengths $\ell_1,\dots,\ell_N \in \R_+^*=(0,+\infty)$, and by $(N+1)$ vertices $v,v_1,\dots,v_N$. We set $\ell = (\ell_1,\dots,\ell_N)$. For each $j \in \NN$ (we denote by $\NN$ the set $\{1,2,\dots,N\}$), the edge $e_j$ is identified with the interval $[0,\ell_j]$, the end of $e_j$ parametrized by $0$ is identified with $v_j$ and the end parametrized by $\ell_j$ is identified with $v$. Then all the egdes have a common end (the vertex $v$). 

A quantum star graph is a metric star graph on which we consider a differential operator. One typically considers the Laplacian on each edge with some boundary conditions at the vertices. In this article, we aim at understanding the spectrum of these operators when we impose boundary conditions of \emph{non-selfadjoint} type. More precisely, let $\a \in \C$. We consider the set $\Dom(H_\a)$ of functions $u = (u_j)_{1 \leq j \leq N}$ belonging to the Sobolev space $\mathcal{H}^2(\Gamma)$ (each $u_j$ belongs to $\mathcal{H}^2(0,\ell_j)$, see \eqref{def-HmGamma} below) and satisfying the following conditions (see Figure \ref{fig-graphe}):

\begin{enumerate} [(i)]
\item $u$ vanishes on $v_1,\dots,v_N$: 
\begin{equation} \label{cond-dirichlet}
\forall j \in \NN, \quad  u_j(0) = 0.
\end{equation}
\item $u$ is continuous at the vertex $v$: 
\begin{equation} \label{cond-continuite}
\forall j,k \in \NN, \quad u_j(\ell_j) = u_k(\ell_k).
\end{equation}
\item At $v$ we have a Robin-type condition on the derivatives of $u$: 
\begin{equation} \label{cond-Kirchhoff}
\sum_{j=1}^N u'_j(\ell_j)  + \a u(v) = 0,
\end{equation}
where $u(v)$ is the common value of $u_j(\ell_j)$, $1 \leq j \leq N$.
\end{enumerate}
Notice that for $\a = 0$ we recover the usual Kirchhoff (or Neumann) condition in \eqref{cond-Kirchhoff}. 
Then, for $\a \in \C$, we define on the domain $\Dom(H_\a)$ the operator $H_\a$ which maps $u = (u_j)_{1 \leq j \leq N}$ to 
\begin{equation} \label{def-Ha}
H_\a u = (-u''_j)_{1 \leq j \leq N}.
\end{equation}

A more complete definition of this operator will be given in Section \ref{sec-def-operator} below. The purpose of this paper is to investigate the spectral properties of $H_\a$. Note that $H_\a$ can also be viewed as a Schr\"odinger type operator with a (possibly complex valued) Dirac potential at the central vertex.

\begin{center}
\begin{figure}[h]
\includegraphics[width = 0.45 \linewidth]{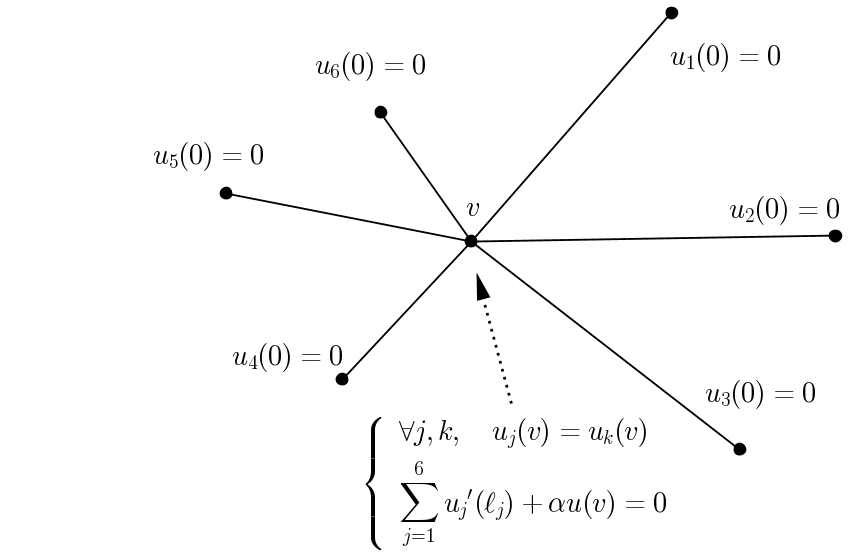}
\caption{The quantum star graph with $N = 6$ edges.}
\label{fig-graphe}
\end{figure}
\end{center}

\subsection{Main results}

It is known (see \cite{Berkolaiko-Kuchment} for a general overview about selfadjoint quantum graphs) that the spectrum of the selfadjoint operator $H_0$ is given by a sequence of positive eigenvalues. We denote them by 
$$
0 < \lambda_1(0)\leq\lambda_2(0)\leq\ldots\leq\lambda_n(0)\leq\ldots
$$
where each eigenvalue is repeated according to its multiplicity. Then, for $n \in \N^*$, we denote by $\tau_n$ the positive square root of $\l_n(0)$. 

We recall (see for instance \cite[Lemma 3.7.4]{Berkolaiko-Kuchment} or Remark \ref{rem-Weyl-H0} below) that we have the following Weyl Law. If for $R_1 < R_2$ we set
$$\mathcal{N}_0(R_1,R_2)= \Card{\{n:R_1<\tau_n\leq R_2\}},$$
then there exists $C_\Gamma>0$ independant of $R_1$ and $R_2$ such that
\begin{equation} \label{eq-Weyl-H0}
\left|\mathcal{N}_0(R_1,R_2)-\frac {\abs\G}\pi (R_2-R_1)\right|\leq C_{\Gamma},
\end{equation}
where we have set
\[
\abs{\G} = \sum_{j=1}^N \ell_j.
\]
This gives in particular 
\begin{equation} \label{eq-Weyl-H0-2}
\frac {\tau_n}{n} \limt n {+\infty} \frac \pi {\abs \G}.
\end{equation}

These properties are also known for any $\a\in\mathbb{R}$, but the corresponding proofs strongly rely on the selfadjointness of $H_\a$. Our first purpose here is to verify that the spectrum of the operator $H_\a$ has in fact similar general properties for any $\a \in \C$. We begin with the nature of the spectrum for a general $\a\in\C$:

\begin{proposition} \label{prop-spectrum-Ha}
For all $\a \in \C$, the spectrum of $H_\a$ consists of infinitely many isolated eigenvalues with finite multiplicities. Moreover, the geometric multiplicity of each eigenvalue coincides with its algebraic multiplicity.
\end{proposition}

We refer to \cite{Hussein14,HusseinKreSie15} for the first spectral properties of general non-selfadjoint Laplace operators on graphs, in particular with maximal sectorial boundary conditions.

Then we turn to the Weyl Law for $H_\a$. The eigenvalues of $H_\a$ are denoted by $\lambda_n(\alpha)$, $n \in \N^*$. They are repeated according to their algebraic multiplicities and ordered with the convention that $$\forall n\in \N^*,\quad\text{Re}(\lambda_n(\alpha))\leq\text{Re}(\lambda_{n+1}(\alpha)).$$
Then, for $n \in \N^*$, we denote by $z_n(\a)$ a square root of $\l_n(\a)$ with non-negative real part. In particular, $z_n(0) = \tau_n$.\\

For $R_1,R_2 \in \R$ with $R_1 < R_2$ we set 
$$
\mathcal{N}_\a(R_1,R_2)= \Card{\{n:R_1< \Re(z_n(\a)) \leq R_2\}}.
$$
Note that the $z_n(\a)$, $n\in\N^*$, are not necessarily ordered by non-decreasing real parts, but we will see that this is in fact the case for $n$ large enough. The next result shows that we recover for any $\a \in \C$ the same Wayl Law as in the selfadjoint case. For $r > 0$, we denote by $D(r)$ the closed disk centered at 0 and with radius $r$ in $\C$.

\begin{proposition}[Weyl Law] \label{th-Weyl}
Let $r > 0$. Then there exists $C > 0$ such that for $\a \in D(r)$ and $R_1,R_2 \in \R$ with $R_1 < R_2$, we have 
\[
\abs{\mathcal{N}_\a (R_1,R_2)-\frac {\abs\G}\pi (R_2-R_1)}\leq C.
\]
\end{proposition}

The proof of Proposition~\ref{th-Weyl} will rely on the fact that for large $n$ the eigenvalue $\l_n(\a)$ is in fact, in a suitable sense, close to $\l_n(0)$. In the following, we aim at describing more precisely the asymptotic behaviour of the difference between $\l_n(0)$ and $\l_n(\a)$. For $n \in \N^*$, we set 
\begin{equation} \label{def-delta}
\d_n(\a) = \l_n(\a) - \l_n(0).
\end{equation}

%
%
%
%
%

Our main result is the following description of the asymptotics of $\d_n(\a)$:

\begin{theorem}\label{t:maintheo2} 
There exists a probability measure $\mu_\ell$ supported in $[0,2|\Gamma|^{-1}]$ such that, for every $\alpha\in\IC$ and for every continuous function $f$ on $\IC$,
$$
\frac 1 n \sum_{k=1}^n f(\d_k(\alpha)) \limt n {+\infty} \int_{0}^{\frac{2}{|\Gamma|}}f(s\alpha) \, d\mu_\ell(s).
$$
Moreover, 
\begin{enumerate} [\rm (i)]
\item if $\ell_j / \ell_1 \in \Q$ for all $j \in \NN$, then $\mu_\ell$ is a linear combination of Dirac masses, one of them being carried by $s=0$;
\item if $\k \cdot \ell \neq 0$ for all $\k \in \Z^N \setminus \set 0$, then the measure $\mu_\ell$ is absolutely continuous with 
respect to the Lebesgue measure and its support is exactly $[0,2|\Gamma|^{-1}]$.
\end{enumerate}
\end{theorem}

Note that $\Im(\l_n(\a)) = \Im(\d_n(\a))$, so the asymptotic properties of $\d_n(\a)$ give in particular the asymptotic properties of the imaginary parts of the eigenvalues of $H_\a$.\\

The measure which appears in Theorem \ref{t:maintheo2} will be defined in terms of the measure introduced by Barra and Gaspard in \cite{barrag00} to describe the asymptotics of $\l_{n+1}(0) - \l_n(0)$. It is in fact more specifically connected to the measure appearing in the related result of Keating, Marklof and Winn on spectral determinants~\cite[Th.~3]{KeatingMarWin03}. As in this reference, it turns out that the size of $\d_n(\a)$ is in fact closely related to the distance between $\l_n(0)$ and the rest of the spectrum of $H_0$. In particular we have the following estimate:

\begin{proposition} \label{prop-eigenvalue-proche}
Let $r > 0$. Then there exists $n_r \in \N^*$ such that for all $\a \in D(r)$ and $n \geq n_r$ we have 
\[
\abs{\d_n(\a)} \leq  \mathsf{dist} \big( \l_n(0), \Sp(H_0) \setminus \set{\l_n(0)} \big).
\]
\end{proposition}

With the Weyl Law \eqref{eq-Weyl-H0}, this implies in particular that the sequence $(\d_n(\a))_{n \in \N^*}$ is bounded (uniformly in $\a \in D(r)$).\\

In the proofs of Proposition \ref{th-Weyl} and Theorem \ref{t:maintheo2}, we will be more precise about the properties of $\d_n(\a)$ in different directions, which we describe in the next results. The following result is already known for $\a$ real:

\begin{proposition} \label{prop-eigenvalue-Dir-alpha}
Assume that there exist $j,k \in \NN$ such that $j \neq k$ and $\ell_j / \ell_k \in \Q$. Let 
\[
\tau \in \frac {\pi \Z}{\ell_j} \cap \frac {\pi \Z}{\ell_k}.
\]
Then $\tau^2$ is an eigenvalue of $H_\a$ for any $\a \in \C$. Moreover, its multiplicity does not depend on $\a$.
\end{proposition}

As we shall see in the proof of Theorem \ref{t:maintheo2}, $0$ is always in the support of the measure $\mu_\ell$. This means that one can 
always find a subsequence of eigenvalues $(\lambda_{n_k}(\a))_{k\geq 1}$ such that $n_k \to +\infty$ and $\d_{n_k}(\alpha)\rightarrow 0.$
We will in fact be slightly more precise on that issue:

\begin{proposition}\label{t:maintheo3} Let $\alpha\in\IC$. 
Then, there exists an increasing sequence $(n_k)_{k \in \N^*}$ in $\N^*$ such that
$$\d_{n_k}(\alpha)= \mathop{\mathcal{O}}_{k \to +\infty} \left( \frac 1 {\l_{n_k}(0)} \right).$$
\end{proposition}

Observe that, under the assumption of Proposition~\ref{prop-eigenvalue-Dir-alpha}, one has in fact $\d_{n_k}(\alpha)=0$ along an increasing subsequence. Hence, Proposition~\ref{t:maintheo3} is meaningful in the case where $\kappa \cdot \ell \neq 0$ for every $\kappa\in\Z^N\setminus\{0\}$. In that case, Theorem~\ref{t:maintheo2} also implies that for every $s\in [0,2|\Gamma|^{-1}]$ there exists a subsequence of eigenvalues 
$(\lambda_{n_k}(\a))_{k\geq 1}$ such that $n_k \to +\infty$ and $\delta_{n_k}(\alpha)\rightarrow s\alpha$. Under appropriate Diophantine assumptions on $\ell$, our proof is slightly more quantitative on that particular question:
\begin{proposition}\label{t:maintheo4} 
There exists $\O \subset (\IR_+^*)^N$ such that $(\IR_+^*)^N \setminus \O$ has null Lebesgue measure and if $\ell \in \O$ then the following assertion holds. 
Let $\e > 0$ and $\a \in \C$. Let $s \in [0,2\abs \G\inv]$. Then there exists an increasing sequence $(n_k)_{k \in \N^*}$ in $\N^*$ such that
$$
\d_{n_k}(\alpha)= s\alpha + \mathop{\mathcal{O}}_{k \to +\infty} \left(\frac 1 {\lambda_{n_k}(0)^{\frac{1}{2N}-\varepsilon}} \right).
$$
\end{proposition}

\subsection{Related literature on non-selfadjoint problems}

The properties of selfadjoint operators on compact graphs are now well understood (see \cite{Berkolaiko-Kuchment}). Less is known on the spectrum of a non-selfadoint Laplacian. We refer to \cite{Hussein14,HusseinKreSie15} for an overview of general properties on Laplace operators with maximal quasi-accretive and general non-selfadjoint vertices conditions.\\

Some authors have also considered the closely related problem of the damped wave equation on graphs, which motivated our analysis.  More precisely, if we consider on $\G$ the problem 
\begin{equation}\label{e:dwe}
\begin{cases}
\partial_t^2 u_j(t,x_j) - \partial_{x_j}^2 u_j(t,x_j) = 0, & \forall t \geq 0, \forall j \in \NN, \forall x_j \in (0,\ell_j),\\
u_j(t,0) = 0 & \forall t \geq 0, \forall j \in \NN,\\
u_j(t,\ell_j) = u_k(t,\ell_k), & \forall t \geq 0, \forall j,k \in \NN,\\
\sum_{j=1}^N \partial_{x_j}u_j(t,\ell_j) + a \partial_t u(t,v) = 0, & \forall t \geq 0,
\end{cases}
\end{equation}
with some initial conditions $(u(0),\partial_t u(0)) = (u_0,u_1) \in \mathcal{H}^1(\G) \times L^2(\G)$. The first results on that direction concern the stabilization of coupled strings joint by a damper (this can be seen as a star-graph with $N=2$ edges). This is related to the accumulation of eigenvalues on the real axis. In \cite{ChenCoWe87} it is shown that, depending on the damper and boundary conditions, we may or may not have exponential decay to equilibrium. In \cite{Liu88,GuoZhu97}, the lengths are still equal to 1 but the wave speeds can be different on the two edges. The asymptotic properties of a wave (the energy does not go to 0, it goes strongly to 0, it decays uniformly exponentially to 0) are then described, depending on the arithmetic properties of the ratio of these wave speeds. Similar stabilization results for the wave equation were also obtained on star graphs~\cite{AmmariJel04, AmmariJelKhe05} (see also \cite{NicaiseVal07}) and on more general graphs~\cite{ValeinZua09, Zuazua13}.

The spectrum of the damped wave equation has been studied in \cite{AbdallahMerNic13} (star-graph) and \cite{FreitasLi17} (graph with general topology). In both cases, the lengths of the edges are rationally dependant. As in the present paper, the result in \cite{FreitasLi17} is the description of the distribution of high frequency eigenvalues. However, the damping is on each edge, and not at the vertices as in the present work or as in~\eqref{e:dwe}. In these references, it is proved that this asymptotics only depends on the mean value of the damping on each edge. Moreover, there are only a finite number of high frequency spectral abscissa. With the notation of our Theorem \ref{t:maintheo2}, this means that the measure $\mu_\ell$ is given by a finite sum of Dirac measures. This is exactly what we recover with a Robin vertex condition for rationnaly dependant lengths (case (i) of Theorem \ref{t:maintheo2}). We also refer to \cite{GuoXu09} for a discussion about the basis properties of the (generalized) eigenfunctions.

We expect our study to be a first a step towards the understanding of the spectrum of the damped wave equation~\eqref{e:dwe} on graphs. Here we consider a Laplacian with a fixed Robin coefficient. The operator corresponding to the wave equation with damping at the central vertex would be of the same form with a frequency dependent Robin coefficient. Indeed, if we look for stationary solutions to the damped wave equation~\eqref{e:dwe} of the form $u_j(t,x_j) = e^{-itz} w_j(x_j),$
then $w = (w_j)_{1 \leq j \leq N}$ has to be a solution of the problem
\[
H_{-iaz} w = z^2 w.
\]
This kind of problems are left for further works. Other classical and related non-selfadjoint problems on graphs are the ones concerned with the distribution of resonances on non-compact quantum graphs, see \cite{ExnerLi10,DaviesExLi10, DaviesPu11, LeeZw16, CdVerdiereTr18} and references therein. For the wave equation on some non-compact star-graph, we also refer to \cite{AsselJeKh16}.\\

Our results can also be compared to questions that appear in the context of the Laplace operator on a compact 
Riemannian manifold $(M,g)$. In the selfadjoint case, we can mention the study of the Schr\"odinger operator
$-\Delta_{\mathbb{S}^n}+a,$ on the sphere $\mathbb{S}^n$ endowed with its canonical metric where $a$ is a 
smooth real-valued function. In that case, eigenvalues are organized in a sequence of clusters and the 
deviation between the eigenvalues for $a\equiv 0$ and for a general $a$ was described 
by Weinstein in~\cite{Weinstein77} in terms of the Liouville measure. When $\alpha\in\IR$, 
our Theorem~\ref{t:maintheo2} can be thought as an analogue of this result where the potential would be a Dirac mass 
at the central vertex of the star graph. Yet, as we shall see in our proof, the measure $\mu_\ell$ that 
we obtain in the case of a graph is of completely different nature.

The spectral properties of the damped wave equation have also attracted a lot of attention on compact manifolds. See for instance the results of Sj\"ostrand~\cite{sjostrand00}, Hitrik~\cite{Hitrik02}, Asch-Lebeau~\cite{AschLe03} and Hitrik-Sj\"ostrand~\cite{HitrikSjo04}. Sj\"ostrand proved that the complex 
eigenvalues of the damped wave operator verify a Weyl Law in the high frequency limit and he considered the distribution of the imaginary parts of their square roots. He showed that any limit distribution is supported on an interval $[a_-,a_+]$ where $a_{\pm}$ corresponds to the infimum (resp. supremum) of the average of the (smooth) damping function $-a$ along a typical orbit of the geodesic flow. In particular, in the case of flat tori or of negatively curved manifolds, both quantities are equal to the average of $-a$ against the normalized volume on $M$. In other words, the limit distribution of the imaginary parts is the Dirac mass at $a_+$. For more general manifolds, identifying the limit measure remains an open problem -- see for instance \cite{Hitrik02, AschLe03, HitrikSjo04} for results in the case of the sphere. 
These results of Sj\"ostrand remain true for non-selfadjoint Schr\"odinger operators with a weaker non-selfadjoint 
perturbation, e.g. $\Delta_g+ia$. In the irrational case (case (ii) of Theorem~\ref{t:maintheo2}), our limit measure is supported on an interval and it is in some sense opposite to the case of flat tori and of negatively curved manifolds, where the limit measure is carried by a single point.

\subsection{The case of more general graphs}

In principle, the results presented in this article could be extended to more general graphs that 
are not star-shapped. Yet, we emphasize that such an extension would a priori require to make some 
restriction on the choice of the metric (i.e. on the choice of $\ell$) we put on the graph. 
In fact, to every choice of graph and to every choice of boundary conditions for the Laplace
operator at the vertices of the graph, one can associate an analytic function $F_{\Gamma,H_0}$ 
defined on the $N$-dimensional torus $\IT^N$, where $N$ is the number of edges of the graph. This function 
does not depend on the choice of $\ell$ and the 
vanishing locus of $F_{\Gamma,H_0}$ plays a central role when studying the spectral properties 
of $H_0$~\cite{Berkolaiko-Kuchment}. Our argument is based on ideas that go back to the works of 
Barra and Gaspard studying the level spacings on quantum graphs~\cite{barrag00}. More precisely, it consists in analyzing the ergodic 
properties of the map induced by the flow
$$\varphi^t_{\ell}:y\in\IT^N\mapsto y+t\ell$$
on the vanishing locus of $F_{\Gamma,H_0}$. To exploit this idea of Barra and Gaspard, it is necessary that the orbit 
of $0$ under this flow intersects the vanishing locus of $F_{\Gamma,H_0}$ on a nice enough subset, 
typically a smooth submanifold. For instance, this can be achieved by considering a quantum graph $H_0$ with Neumann 
boundary conditions at each vertex (for which $F_{\Gamma,H_0}$ is nice enough thanks to~\cite{Friedlander05}) and by picking a \emph{generic} metric $\ell$ (typically such 
that $\varphi_{\ell}^t$ is uniquely ergodic). This was the approach followed in earlier works of Berkolaiko-Winn on level spacings~\cite{BerkolaikoWi10} or of Colin de Verdi\`ere on the failure of quantum ergodicity for quantum graphs~\cite{cdverdiere15}. See also~\cite{KeatingMarWin03, BerkolaikoKeaWin04} for earlier results on these problems by Keating et al. in the case of star graphs (still for a generic $\ell$). Here, we chose to work only on star-graphs for which the expression of $F_{\Gamma,H_0}$ is explicit enough to determine all the singularities of the vanishing locus and to handle \emph{any} choice of metric $\ell$. This will illustrate the difference of behaviour depending on our choice of metric $\ell$. 


Finally, in order to extend the results to more general graphs, we would need to define an appropriate equivalent 
of $\d_n(\alpha)$ (for instance, the imaginary part of $\lambda_n(\alpha)$ when $\alpha\notin\mathbb{R}$) in order to state an analogue of Theorem~\ref{t:maintheo2}. Typically, the proof of Proposition \ref{prop-sep-vp} below, which gives Proposition \ref{prop-eigenvalue-proche} and then Theorem \ref{t:maintheo2}, relies on explicit computations.

\subsection{Organization of the paper} In Section \ref{sec-def-operator} we define the operator $H_\a$ and give some basic spectral properties. In particular, we compute the secular equation on which all the subsequent analysis relies. In Section \ref{sec-H0}, we focus on the distribution of the eigenvalues of $H_0$. In particular, we introduce the corresponding dynamical system on $\T^N$ and the Barra-Gaspard measure. In Section \ref{sec-separation} we prove that, for high frequencies, the perturbation of the Robin parameter does not move the spectrum so much. So each eigenvalue of $H_\a$ remains close, in a suitable sense, to the corresponding eigenvalue of $H_0$. This will finish the proof of Proposition \ref{prop-spectrum-Ha} and prove Propositions \ref{th-Weyl},~\ref{prop-eigenvalue-proche} and~\ref{prop-eigenvalue-Dir-alpha}. Finally, the last section is devoted to the proofs of Theorem \ref{t:maintheo2} and Propositions \ref{t:maintheo3} and \ref{t:maintheo4}.

\subsection{Notation} We denote by $\N$ the set of non-negative integers, and $\N^*$ stands for $\N \setminus \set 0$. As already mentioned, we also write $\NN$ for $\{1,\cdots,N\}$. For $z_0 \in \C$ and $r > 0$ we set $D(z_0,r) = \set{z \in \C : \abs{z-z_0} \leq r}$ and $C(z_0,r) = \set{z \in \C : \abs{z-z_0} = r}$.  As said above, we write $D(r)$ for $D(0,r)$.

\section{Spectral framework and first properties} \label{sec-spectral}

In this section, we define the operator $H_\a$ on a suitable Hilbert space on $\G$ (we recall that Laplace operators with more general non-selfadjoint vertices conditions are discussed in \cite{Hussein14,HusseinKreSie15}). Then we derive some basic properties for its spectrum. In particular, we give some rough localization for the eigenvalues, and we derive the secular equation, that will be used all along the article.

\subsection{The Schr\"odinger operator with Robin condition} \label{sec-def-operator}
The metric graph $\G$ is endowed with the Hilbert structure of
\[
L^2(\Gamma) = \bigoplus _{j=1}^N L^2(0, \ell_j).
\]
In particular, for $u = (u_j)_{1 \leq j \leq N} \in L^2(\Gamma)$, we have 
\[
\nr{u}_{L^2(\G)}^2 = \sum_{j=1}^N \nr{u_j}_{L^2(0,\ell_j)}^2. 
\]
For $m \in \N$, we similarly define the Sobolev space
\begin{equation} \label{def-HmGamma}
\mathcal{H}^m(\Gamma) = \bigoplus _{j=1}^N \mathcal{H}^m \left(0, \ell_j \right).
\end{equation}
In the introduction, we have denoted by $\Dom(H_\a)$ the set of $u \in \mathcal{H}^2(\G)$ such that \eqref{cond-dirichlet}, \eqref{cond-continuite} and \eqref{cond-Kirchhoff} hold, and for $u \in \Dom(H_\a)$ we have defined $H_\a u \in L^2(\G)$ by \eqref{def-Ha}. For later purpose, we also define $\Dom(H_\infty)$ as the set of $u = (u_j)_{1 \leq j \leq N} \in \mathcal{H}^2(\G)$ such that 
\begin{equation}\label{cond-full-Dirichlet}
\forall j \in \NN, \quad u_j(0) = u_j(\ell_j) = 0.
\end{equation}
Then, for $u \in \Dom(H_\infty)$, we define $H_\infty u \in L^2(\G)$ as in \eqref{def-Ha}.\\

It is known that, for $\a \in \R \cup \set \infty$, the operator $H_\a$ is selfadjoint and bounded from below (it is even non-negative for $\a \geq 0$ or $\a = \infty$). In particular, since $\mathcal{H}^2(\G)$ is compactly embedded in $L^2(\G)$, the spectrum of $H_\a$ consists of a sequence of isolated real eigenvalues going to $+\infty$. When $\a$ is not real, the operator $H_\a$ is no longer symmetric. However, we prove in Proposition \ref{prop-sectorial} that $H_\a$ is at least maximal sectorial. We recall that an operator $T$ on some Hilbert space $\Hc$ is said to be sectorial (with vertex $\g \in \R$ and semi-angle $\th \in [0,\frac \pi 2[$) if the numerical range of $T$,
\[
\set{\innp{T\f}{\f}_\Hc, \f \in \Dom(T), \nr{\f}_\Hc = 1},
\]
is a subset of the sector 
\[
\Sigma_{\g,\th} = \set{\z \in \C \st \abs{\arg(\z-\g)} \leq \th}.
\]
Then $T$ is said to be maximal sectorial if $(T-\z)$ has a bounded inverse on $\Hc$ for some (and hence any) $\z$ in $\C \setminus \Sigma_{\g,\th}$. Similarly, we say that a quadratic form $q$ on $\Hc$ with domain $\Dom(q)$ is sectorial (with vertex $\g \in \R$ and semi-angle $\th \in [0,\frac \pi 2[$) if for all $\f \in \Dom(q)$ with $\nr{\f}_\Hc = 1$ we have 
\[
q(\f) \in \Sigma_{\g,\th}.
\]

In this paper, we study the properties of $H_\a$ for some fixed $\a \in \C$. However, at some point (see the proof of Proposition \ref{prop-zna} below), we will deduce some properties of $H_\a$ from the corresponding properties of $H_0$ by a regularity argument with respect to the \emph{parameter} $\a$. This is why we need uniform estimates on $H_{\tilde \a}$ for $\tilde \a$ lying on a disk containing $0$ and $\a$. Thus, we fix $r > 0$ and give the following results with constants uniform with respect to $\a$ in the disk $D(r)$.

\begin{proposition} \label{prop-sectorial}
Let $r > 0$. There exist $\g \in \R$ and $\th \in \big[0,\frac \pi 2 \big[$ such that for any $\a \in D(r)$ the operator $H_\a$ is maximal sectorial with vertex $\g$ and semi-angle $\th$.
\end{proposition}

\begin{proof} 
\stepp 
Let $\a \in D(r)$. We denote by $\Dom(Q_\a)$ the set of functions $u \in \mathcal{H}^1(\G)$ which satisfy \eqref{cond-dirichlet} and \eqref{cond-continuite} (notice that this set does not depend on $\a$). Then, for $u,w \in \Dom(Q_\a)$, we set 
\[
Q_\a(u,w) = \sum_{j=1}^N \int_0^{\ell_j} u'_j(x_j) \overline{w'_j(x_j)} \, dx_j + \a u(v) \overline{w(v)}.
\]
In particular, $\Dom(H_\a) \subset \Dom(Q_\a)$ and, for $u \in \Dom(H_\a)$ and $w \in \Dom(Q_\a)$, we have
\begin{equation} \label{eq-Ha-Qa}
\innp{H_\a u}{w}_{L^2(\G)} = Q_\a (u,w).
\end{equation}
\stepp For $u \in \Dom(Q_\a) = \Dom(Q_0)$ and $j \in \NN$ we have 
\begin{align*}
\abs{(Q_\a -Q_0)(u)} 
 = \abs{\a} \abs{u(v)}^2
& \leq 2\abs \a \int_0^{\ell_j} \abs{u_j(x_j)} \abs {u_j'(x_j)} \, dx_j\\
& \leq  \frac{1}{4}\nr{u'_j}_{L^2(0,\ell_j)}^2 + 4\abs \a^2 \nr{u_j}_{L^2(0,\ell_j)}^2\\
& \leq \frac{1}{4}Q_0(u) +  4 r^2 \nr u_{L^2(\G)}^2.
\end{align*}
Since the form $Q_0$ is symmetric, bounded from below and closed, we deduce that the form $Q_\a$ is sectorial (uniformly in $\a \in D(r)$) and closed~\cite[Th.~VI.1.33]{kato}.

\stepp By the Representation Theorem~\cite[Th.~VI.2.1]{kato}, there exists a unique maximal sectorial operator $\tilde H_\a$ on $L^2(\G)$ such that $\Dom(\tilde H_\a) \subset \Dom(Q_\a)$ and 
\[
\forall u \in \Dom(\tilde H_\a), \forall w \in \Dom(Q_\a), \quad \big< \tilde H_\a u , w \big>_{L^2(\G)} = Q_\a (u,w).
\]
Moreover, $\Dom(\tilde H_\a)$ is the set of $u \in \Dom(Q_\a)$ for which there exists $f \in L^2(\G)$ such that 
\[
\forall w \in \Dom(Q_\a), \quad Q_\a(u,w) = \innp{f} {w}_{L^2(\G)},
\]
and, in this case, $\tilde H_\a u = f$.

\stepp By \eqref{eq-Ha-Qa}, we have $\Dom(H_\a) \subset \Dom(\tilde H_\a)$ and $\tilde H_\a$ coincides with $H_\a$ on $\Dom(H_\a)$. Now let $u = (u_j)_{1 \leq j \leq N} \in \Dom(\tilde H_\a)$. Let us prove that $u \in \Dom(H_\a)$. There exists $f = (f_j)_{1 \leq j \leq N}\in L^2(\G)$ such that, for all $w \in \Dom(Q_\a)$, we have 
\[
Q_\a (u,w) = \innp f w_{L^2(\G)}.
\]
Let $k \in \NN$. Considering all the test functions $w = (w_j)_{1\leq j \leq N}$ such that $w_k \in C_0^\infty((0,\ell_j))$ and $w_j = 0$ for $k \neq j$, we obtain that $u_k$ belongs to $\mathcal{H}^2(0,\ell_k)$ and $-u''_k = f_k$. Then $u \in \mathcal{H}^2(\G)$. Now, for all $w \in \Dom(Q_\a)$, we have 
\begin{align*}
Q_\a (u,w)
& = \sum_{j = 1}^N \innp{u'_j}{w'_j}_{L^2(0,\ell_j)} + \a u(v) \overline {w(v)}\\
& = \sum_{j=1}^N \left( u_j'(\ell_j) \overline{w_j(\ell_j)} - \innp{u''_j}{w_j}_{L^2(0,\ell_j)} \right) + \a u(v) \overline {w(v)}\\
& = \innp{f}{w}_{L^2(\G)} + \left( \sum_{j=1}^N u_j'(\ell_j) + \a u(v) \right) \overline{w(v)}.
\end{align*}
Choosing $w \in \Dom(Q_\a)$ such that $w(v) \neq 0$ yields \eqref{cond-Kirchhoff}. This proves that $u \in \Dom(H_\a)$, and hence $H_\a = \tilde H_\a$ is maximal sectorial.
\end{proof}

\begin{remark}
With the same proof we recover that $H_\a$ is selfadjoint if $\a$ is real, and non-negative if $\a \geq 0$ (see~\cite[Th.~1.4.11]{Berkolaiko-Kuchment} for a similar result for more general selfadjoint vertices conditions and more general graphs).
\end{remark}

\subsection{First properties of the spectrum} \label{sec-first-properties}

In this paragraph we begin our description of the localization of the spectrum with some rough properties.

\begin{proposition}\label{prop-spectre-discret}
Let $r > 0$. Let $\g$ and $\th$ be given by Proposition \ref{prop-sectorial}. Let $\a \in D(r)$. Then the spectrum of the operator $H_\a$ is included in $\Sigma_{\g,\th}$ and consists of isolated eigenvalues with finite multiplicities.
\end{proposition}

\begin{proof}
By Proposition \ref{prop-sectorial}, the operator $H_\a$ is maximal sectorial with vertex $\g \in \R$ and semi-angle $\th \in \big(0,\frac \pi 2\big)$, and hence its spectrum is a subset of $\Sigma_{\g,\th}$~\cite[\S V.3.10]{kato}. In particular, the resolvent set of $H_\a$ is not empty. Moreover, $\Dom(H_\a)$ is compactly embedded in $L^2(\G)$, so the operator $H_\a$ has a compact resolvent. Thus, its spectrum consists of isolated eigenvalues with finite multiplicities.
\end{proof}

Notice that we have not proved yet that $H_\a$ has an infinite number of eigenvalues. We will get this as a byproduct of the more refined analysis of Section \ref{sec-separation} (see also \cite{HusseinKreSie15}).

\begin{remark} Even if we are mostly interested in the large eigenvalue limit of the spectrum, we note that $0$ is an eigenvalue of $H_\a$ if and only if 
\begin{equation*} 
\a = \sum_{j=1}^N \frac 1 {\ell_j},
\end{equation*}
and that in this case $\ker(H_\a)$ is generated by $u = (u_j)_{1 \leq j \leq N}$ such that $u_j(x_j) = x_j / \ell_j$ for all $j \in \NN$ and $x_j \in [0,\ell_j]$.\\
\end{remark}

Now we consider non-zero eigenvalues. Let $z \in \C^*$. For $u = (u_j)_{1 \leq j \leq N} \in \mathcal{H}^2(\G)$, we have $-u_j'' = z^2 u_j$ 
for all $j \in \NN$ if and only if there exist constants $(\b_j,\tilde \b_j)_{1 \leq j \leq N}$ such that, for $j \in \NN$ and for $x \in [0,\ell_j]$, 
\begin{equation*}
u_j(x) = \b_j \sin(zx) + \tilde \b_j \cos(zx).
\end{equation*}
Such a function $u$ belongs to $\Dom(H_\a)$ (and then $H_\a u = z^2 u$) if and only if the vertices conditions \eqref{cond-dirichlet}, \eqref{cond-continuite} and \eqref{cond-Kirchhoff} are satisfied. The Dirichlet condition \eqref{cond-dirichlet} at the external vertices gives 
\begin{equation} \label{eq-tilde-beta-j}
\forall j \in \NN, \quad \tilde \b_j = 0,
\end{equation}
so that 
\begin{equation} \label{expr-u-beta}
u_j(x) = \b_j \sin(zx).
\end{equation}

With this simple calculation, we can already prove the following rough result.

\begin{proposition} \label{prop-strip}
\begin{enumerate}[\rm (i)]
\item Let $\a \in \C$. If $\l \in \C$ is an eigenvalue of $H_\a$ then $\Im(\l)$ and $\Im(\a)$ have the same sign (recall that the spectrum of $H_\a$ is real if $\a$ is real, on the other hand there can be real eigenvalues even if $\a$ is not real).
\item Let $r > 0$. There exists $c_r > 0$ such that, if $\l \in \C$ is an eigenvalue of $H_\a$ for some $\a \in D(r)$, then 
\[
\Re(\l) \geq -c_r \quad \text{and} \quad \abs{\Im(\l)} \leq c_r.
\]
Moreover, we have 
$$
\limsup_{R \to +\infty} \sup_{\substack{\l \in \Sp(H_\a)\\ \Re(\l) \geq R}} |\Im(\l)| \leq \frac {2\abs{\Im(\a)}} {|\Gamma|}.
$$
\end{enumerate}
\end{proposition}

\begin{remark}
The last statement should be compared with the support of $\mu_\ell$ in Theorem \ref{t:maintheo2}. In particular, this is an equality when $\kappa \cdot \ell\neq 0$ for every $\kappa\in\Z^N\setminus\{0\}$.
\end{remark}

\begin{proof}
Let $z = \tau + i\eta \in \C^*$ (with $\tau,\eta \in \R$) and assume that $\l = z^2$ is an eigenvalue of $H_\a$. Let $u \in \ker(H_\a - z^2) \setminus \set 0$, normalized by $\nr{u}_{L^2(\G)} = 1$. Let $\b_1,\dots,\b_N \in \C$ be such that \eqref{expr-u-beta} holds. Taking the imaginary parts in the equality 
\[
\l = \innp{H_\a u}{u} = Q_\a(u,u)
\]
gives
\begin{equation}\label{e:sign}
\Im(\l) = \Im(\a) \abs{u(v)}^2.
\end{equation}
This implies in particular that $\Im(\l)$ and $\Im(\a)$ have the same sign.

Then we compute
\begin{equation}\label{eq-nr-betaj}
\begin{aligned}
1 = \nr{u}_{L^2(\G)}^2 = \sum_{j=1}^N \int_0^{\ell_j} \abs{\b_j \sin(z x_j)}^2 \, dx_j
= \sum_{j=1}^N \abs{\b_j}^2 \left( \frac {\sinh(2\eta \ell_j)}{4\eta} - \frac {\sin(2 \tau \ell_j)}{4\tau} \right).
\end{aligned}
\end{equation}

Let $\seq \l m$ be a sequence of complex numbers such that $\Re(\l_m)$ goes to $+\infty$ and, for all $m \in \N$, $\l_m$ is an eigenvalue of $H_{\a_m}$ for some $\a_m \in D(r)$. For $m \in \N$ we consider $\tau_m \geq 0$ and $\eta_m \in \R$ such that $\l_m = (\tau_m + i\eta_m)^2$. Then $\tau_m$ goes to $+\infty$ as $m$ goes to $+\infty$. Assume for instance that $\Im(\a) \geq 0$. After extracting a subsequence if necessary, we can assume that the sequence $\seq \eta m$ has a limit $\eta_\infty$ in $[0,+\infty]$. For $m \in \N$ we consider $u_m \in \ker(H_{\a_m}-\l_m)$ such that $\nr{u_m}_{L^2(\G)} = 1$ and we denote by $\b_{1,m},\dots,\b_{N,m} \in \C$ the corresponding coefficients in \eqref{eq-tilde-beta-j}.

Assume that $\eta_\infty = + \infty$. Then, by \eqref{eq-nr-betaj}, we have 
\[
\sum_{j=1}^N  \frac {\abs{\b_{j,m}}^2 e^{2\eta_m \ell_j}}{4\eta_m} \limt m {+\infty} 1.
\]
Then
\[
\abs{u(v)}^2 = \abs{\b_{1,m}}^2 \abs{\sin(z_m \ell_1)}^2 \leq \abs{\b_{1,m}}^2 e^{2\eta_m \ell_1} = \Oc (\eta_m),
\]
and hence, by \eqref{e:sign},
\[
2 \tau_m \eta_m = \Im(\l_n) \leq r \abs{u(v)}^2 =  \Oc(\eta_m).
\]
This gives a contradiction.

We similarly get a contradiction if we assume that $\eta_\infty \in (0,+\infty)$ so, finally, $\eta_m$ goes to 0 as $m$ goes to $+\infty$. \eqref{eq-nr-betaj} now gives 
\[
\sum_{j=1}^N  \frac {\abs{\b_{j,m}}^2 \ell_j}{2} \limt m {+\infty} 1
\]
so 
\[
\limsup_{m \to +\infty} \frac {\abs \G \abs{u_m(v)}^2}2 = \limsup_{m \to +\infty} \sum_{j=1}^N \frac {\abs {\b_{j,m}}^2 \ell_j} 2 \abs{\sin(z_m \ell_j)}^2 \leq 1.
\]
Then \eqref{e:sign} gives
\[
\limsup_{m \to +\infty} \abs{\Im(\l_n)} \leq \abs{\Im(\a)} \abs {u_m(v)}^2 \leq \frac {2 \abs {\Im(\a)}}{\abs \G}. 
\]
We conclude with Proposition \ref{prop-spectre-discret}.
\end{proof}


\subsection{Comparison with the eigenvalues of the Dirichlet problem} \label{sec-eigenvalue-Dirichlet}

The spectrum of $H_\infty$ (see \eqref{cond-full-Dirichlet}) is completely explicit. If we set 
\[
\Tc_\infty = \bigcup_{j=1}^N \frac {\pi \Z}{\ell_j},
\]
then 
\[
\Sp(H_\infty) = \set{\tau^2, \tau \in \Tc_\infty}.
\]
Moreover, if $\tau^2$ is an eigenvalue of $H_\infty$, then its multiplicity is given by
$$\Card{\left\{j \in \NN : \tau \ell_j \in \pi \Z\right\}}.$$
In particular, if the lengths of the edges are pairwise incommensurable, then all the eigenvalues of $H_\infty$ are simple. On the contrary, if all the lengths are equal, then all the eigenvalues of $H_\infty$ have multiplicity $N-1$.

In the selfadjoint case ($\a \in \R$) it is known (see~\cite[Th.~3.1.8]{Berkolaiko-Kuchment} or Remark \ref{rem-Weyl-H0} below) that between two consecutive eigenvalues of $H_\infty$ there is exactly one simple eigenvalue of $H_\a$. In the degenerate situation where $\l \in \R$ is an eigenvalue of $H_\infty$ of multiplicity $m \geq 2$, then $\l$ is for all $\a \in \R$ an eigenvalue of multiplicity $m-1$ for $H_\a$. The first statement has no meaning for the possibly non-real eigenvalues of $H_\a$, $\a \notin \R$, but the second still holds in general.

\begin{lemma}\label{l:multiplicity} Let $\alpha\in\IC$.
 \begin{enumerate}[\rm (i)]
  \item If $\tau\in\Tc_\infty$, then
  \[
\dim ( \ker( H_\a - \tau^2) ) = \dim (\ker (H_\infty - \tau^2)) - 1.
\]
 \item If $z\in\IC\setminus\Tc_\infty$, then
$$\dim(\ker(H_\a-z^2)) \leq 1.$$
 \end{enumerate}

\end{lemma}

\begin{proof}
Assume that $\tau^2$ is an eigenvalue of $H_\infty$. There exists $j \in \NN$ such that $\tau \ell_j / \pi$ is an integer. Let $u \in \ker(H_\a - \tau^2)$. By \eqref{expr-u-beta} we have $u(v) = u_j(\tau \ell_j) = 0$. This implies that $\ker(H_\a - \tau^2) \subset \ker(H_\infty-\tau^2)$. Moreover, 
the map $u \mapsto \sum_{j = 1}^N u_j'(\ell_j)$ is a non-zero linear form on $\ker(H_\infty - \tau^2)$, which yields the first part of the lemma.

Now let $z \in \C \setminus \Tc_\infty$. If $u \in \ker(H_\a-z^2)$ is such that $u(v) = 0$, then it belongs to $\ker(H_\infty-z^2)$, so $u = 0$. This means that the linear form $u \mapsto u(v)$ is injective on $\ker(H_\a-z^2)$ and proves the second statement.
\end{proof}

\subsection{The secular equation} \label{sec-secular}

In Paragraph \ref{sec-first-properties}, we have only used the vertex condition \eqref{cond-dirichlet}. For a more refined analysis, we have to take into account \eqref{cond-continuite} and \eqref{cond-Kirchhoff}.

Let $z \in \C^*$ and $u \in \ker(H_\a-z^2)$. Let $\b \in \C^N$ be such that \eqref{expr-u-beta} holds. Then the conditions \eqref{cond-continuite} and \eqref{cond-Kirchhoff} at the vertex $v$ read
\begin{equation} \label{eq-continuite-z}
 \b_1 \sin(z \ell_1) =\b_2 \sin(z \ell_2)=\ldots= \b_N \sin(z\ell_N)
\end{equation}
and
\begin{equation} \label{eq-Kirchhoff-z}
z \sum_{j=1}^N  \b_j \cos(z\ell_j) + \a \b_N \sin(z \ell_N) = 0.
\end{equation}
We can divide this last equality by $z$. Conversely, if $\b$ satisfies \eqref{eq-continuite-z}-\eqref{eq-Kirchhoff-z} and $u \in L^2(\G)$ is defined by \eqref{expr-u-beta} then $u \in \Dom(H_\a)$ and $H_\a u = z^2 u$. Thus, $z^2\neq 0$ is an eigenvalue of $H_\a$ if and only if
\[
\det \left( A\left( z\ell, \frac \a z \right) \right)= 0,
\]
where, for $y= (y_1,\dots,y_N) \in \C^N$ and $\eta \in \C$, 
\[
A(y,\eta) =\left(\begin{array}{ccccc}\sin(y_1) & -\sin(y_2) & 0 &\ldots & 0\\
                                                 0        &\sin(y_2)  &-\sin(y_3) &0&\ldots\\
                                                 \vdots & \ddots &\ddots &\ddots &\vdots\\
                                                 0 &\ldots &0 &\sin(y_{N-1}) &-\sin(y_N)\\
                                                 \cos(y_1) &\cos(y_2) &\ldots &\ldots& \cos(y_N)+\eta\sin(y_N)
                             \end{array}
                       \right).
\]
Moreover, since the map 
\[
(\b_j)_{1 \leq j \leq N} \in \ker \left( A \left( z\ell, \frac \a z \right) \right) \mapsto u \in \ker\left(H_\a-z^2 \right),
\]
defined by \eqref{expr-u-beta}, is an isomorphism, we have
\[
\dim \left( \ker \left( A \left( z\ell, \frac \a z \right) \right) \right) = \dim\left(\ker\left(H_\a-z^2 \right)\right).
\]
By a straightforward computation we observe that 
\begin{equation}\label{e:det-robin}
\det(A(y,\eta)) = F_N(y) + \eta F_D(y),
\end{equation}
where 
\begin{equation} \label{e:det-neumann-explicit}
F_N(y_1,\ldots,y_N)=\sum_{j=1}^N\cos(y_j)\prod_{k\neq j}\sin(y_k)
\end{equation}
and
\begin{equation} \label{e:det-dirichlet}
F_D(y_1,\ldots,y_N)=\prod_{j=1}^N\sin(y_j).
\end{equation} 

Note that $F_N = \det(A(\cdot,0))$ is the determinant associated with the Neumann (or Kirchhoff) problem ($\a = 0$) and $F_D$ is the determinant corresponding to the Dirichlet problem, i.e. $z^2$ is an eigenvalue of $H_\infty$ if and only if $F_D(z\ell) = 0$.\\

For $y \in \C^N$ such that $F_D(y) \neq 0$ we set 
\begin{equation}\label{e:neumann-cotan}
\Psi(y)= - \frac {F_N(y)}{F_D(y)} = -\sum_{j=1}^N\cotan(y_j).
\end{equation}

\begin{remark} \label{rem-Weyl-H0}
Let $\tau \in \R_+ \setminus \Tc_\infty$. Then $\tau^2$ is an eigenvalue of $H_0$ is and only if $\Psi(\tau \ell) = 0$. We have 
\[
\frac d {d\tau} \Psi(\tau \ell) = \sum_{j=1}^N \ell_j \big( 1 + \cotan(\tau \ell_j)^2 \big) > 0.
\]
Since $\abs{\Psi(\tau\ell)}$ goes to $+\infty$ when $\tau$ approaches $\Tc_\infty$, we see that in each connected component $I$ of $\R_+ \setminus \Tc_\infty$ (which does not contain 0) there exists exactly one $n \in \N^*$ such that $\tau_n \in I$. See Figure \ref{fig-Psi} in Section \ref{sec-separation} below. On the other hand, we deduce from the discussion of Paragraph \ref{sec-eigenvalue-Dirichlet} that the operator $H_\infty$ satisfies a Weyl Law as in \eqref{eq-Weyl-H0}. Then, combining Lemma \ref{l:multiplicity} and the previous remark, we recover \eqref{eq-Weyl-H0} for $H_0$. In fact, the same applies to $H_\a$ for any $\a \in \R$ if we observe that on $\R_+ \setminus \Tc_\infty$ we also have 
\[
\frac d {d\tau} \big( \tau \Psi(\tau \ell) \big) > 0.
\]

\end{remark}

\section{Eigenvalues of \texorpdfstring{$H_0$}{H0} and the Barra-Gaspard measure} \label{sec-H0}

In this section, we review a few facts on a natural Radon measure associated to our metric star graph $(\Gamma,\ell)$. This measure was defined by Barra and Gaspard in~\cite{barrag00} to study the level spacings of the operator $H_0$ (see also~\cite{BerkolaikoWi10}) and further used to study the distribution of the eigenmodes of $H_0$~\cite{KeatingMarWin03, BerkolaikoKeaWin04, cdverdiere15} or resonances on noncompact graphs~\cite{CdVerdiereTr18}. The main differences with these references is that we use the explicit structure of the graph to handle \emph{any} choice of metric on $\Gamma$.

\subsection{A stratified manifold associated to the eigenvalues of \texorpdfstring{$H_0$}{H0}}

We have said in the previous section that $\tau^2$ is an eigenvalue of $H_0$ if and only if $F_N(\tau \ell) = 0$, where $F_N$ is defined by \eqref{e:det-neumann-explicit}. 
We set 
\[
Z = \set{y \in \R^N : F_N(y) = 0}.
\]
This defines a stratified submanifold, which can be splitted as
$$
Z=\bigsqcup_{J\subset \NN}Z_J,
$$
where, for $J \subset \NN$,
$$
Z_J=\left\{y\in Z: y_j= 0\ \text{mod}\ \pi\Leftrightarrow\ j\in J\right\}.
$$

With this definition of $Z$, we see that $\tau^2$ is an eigenvalue of $H_0$ if and only if $\tau \ell \in Z$. More precisely,
\begin{enumerate}[\rm (i)]
\item if $\tau \ell \in Z_\emptyset$ then $F_D(\tau \ell) \neq 0$ so, by Lemma \ref{l:multiplicity}, $\tau^2$ is a simple eigenvalue of $H_0$ ;
\item if $\tau \ell \in Z_J$ for some $J \subset \NN$ with $\Card J \geq 2$, then $\tau^2$ is an eigenvalue of $H_\infty$ of multiplicity $\Card J$ and, again by Lemma \ref{l:multiplicity}, $\tau^2$ is an eigenvalue of $H_0$ of multiplicity $\Card J - 1$.
\end{enumerate}
Note that $Z_J$ is empty if $\Card J=1$ and $Z_\emptyset = \Psi\inv(\set 0)$ (see \eqref{e:neumann-cotan}). Then for $J \subset \NN$ we set 
\begin{equation} \label{def-mJ}
m_J = 
\begin{cases}
1 & \text{if } J = \emptyset,\\
\Card J-1 & \text{if } \Card J \geq 2.
\end{cases}
\end{equation}

Now we check that each stratum of $Z$ is a submanifold of $\R^N$ to which the vector $\ell$ is transverse.

\begin{lemma} \label{lem-Z-manifold}
\begin{enumerate}[\rm (i)]
\item $Z_\emptyset$ is a submanifold of $\R^N$ of dimension $N-1$, and for all $y \in Z_\emptyset$ we have $\nabla \Psi(y) \cdot \ell \neq 0$.
\item Let $J \subset \NN$ with $\Card J \geq 2$. Then $Z_J$ is a submanifold of $\R^N$ of dimension $N - (\Card J)$ and its boundary $\partial Z_J = \overline{Z_J} \setminus Z_J$ satisfies
\[
\partial Z_J = \bigcup_{J \subsetneqq J' \subset \NN} Z_{J'}.
\]
Moreover, $\ell$ is transverse to $Z_J$.
\end{enumerate}
\end{lemma}

\begin{proof}
\stepp Let $y \in Z_\emptyset$. We have
\[
\nabla \Psi(y) \cdot \ell = \sum_{j=1}^N \ell_j \big( 1 + \cotan(y_j)^2 \big) > 0.
\]
In particular, $\nabla \Psi(y) \neq 0$. By the Implicit Function Theorem, this proves that $Z_\emptyset = \Psi\inv(\set 0)$ is a submanifold of dimension $(N-1)$ in $\R^N$, and $\ell$ is transverse to $Z_\emptyset$.

\stepp Now we consider $J \subset \NN$ with $\Card J \geq 2$. Then any $y \in \R^N$ such that $y_j = 0 \, \text{mod}\ \pi$ for all $j \in J$ belongs to $Z$, so 
\[
Z_J = \left\{y\in \R^N : y_j= 0\ \text{mod}\ \pi \, \Leftrightarrow \, j\in J\right\}.
\]
This defines a submanifold of dimension $N - (\Card J)$. On the other hand, it is not difficult to check that  
\[
\overline{Z_J} = \left\{y\in \R^N : y_j= 0\ \text{mod}\ \pi, \, \forall j\in J\right\},
\]
and the statement about $\partial Z_J$ follows. Finally, $\ell$ is transverse to each $Z_J$ because all its components are positive.
\end{proof}

%

\subsection{Restriction to a dynamical system on a torus}

Since $F_N$ is $2\pi$-periodic with respect to each variable, the condition $F_N(\tau \ell) = 0$ can be seen as an equation on the torus $\T^N = \R^N / (2\pi\Z^N)$.
We set
$$\Lambda_\ell=\left\{k\in\IZ^N:k \cdot \ell=0\right\}$$
and
$$\Lambda_\ell^\perp=\left\{\xi\in\IR^N:\forall k\in\Lambda_\ell,\ k \cdot \xi=0\right\}.$$
In particular, $\tau \ell$ belongs to $\L_\ell^\perp$ for all $\tau \in \R$. We denote by $N_\ell$ the dimension of $\L_\ell^\perp$.

\begin{remark}
$N_\ell$ is the smallest integer for which there exist $\tilde \ell_1,\dots,\tilde \ell_{N_\ell} \in (0,+\infty)$ such that for all $j \in \NN$ we have 
\[
\ell_j = \sum_{k=1}^{N_\ell} m_{j,k} \tilde \ell_k,
\]
for some $m_{j,1},\dots,m_{j,N_\ell} \in \N^*$.
\end{remark}

Noticing that $2\pi\IZ^N\cap\Lambda_\ell^\perp$ is a lattice of $\L_\ell^\perp$, we define the torus
$$
\IT_\ell=\Lambda_\ell^\perp/(2\pi\IZ^N\cap\Lambda_\ell^\perp).
$$ 
It can be identified with a subset of $\T^N$. We denote by $\m_{\T_\ell}$ the Lebesgue measure on $\T_\ell$ (inherited from the Lebesgue measure on $\L_\ell^\perp$) and we set $\abs{\T_\ell} = \m_{\T_\ell}(\T_\ell)$.\\

For $y \in \T^N$ and $t \geq 0$ we set 
\[
\f_\ell^t (y) = y + t\ell \quad \text{mod } 2 \pi.
\]
From the unique ergodicity of this flow on $\T_\ell$, one has, for every continuous function $f$ on $\IT^N$ and uniformly for $y\in\IT_\ell$,
\begin{equation}\label{e:unique-ergodicity}
 \lim_{T\rightarrow+\infty}\frac{1}{T}\int_0^T (f\circ\varphi_\ell^t)(y)dt
 =\frac{1}{\abs{\IT_\ell}} \int_{\IT_\ell}f \, d\m_{\T_\ell}
\end{equation}
(this can be computed explicitely if $f$ is of the form $y \mapsto e^{i \kappa \cdot y}$ for some $\kappa \in \Z^N$, and the general case follows by decomposing $f$ in Fourier series). This implies in particular that any orbit of the flow $\varphi^t_\ell$ starting in $\T_\ell$ is dense in $\T_\ell$.\\

Since $Z$ (and each $Z_J$, $J \subset \NN$) is $(2\pi\Z^N)$-periodic, we can consider 
$$\mathcal{Z}=(Z \cap \Lambda_\ell^\perp)/(2\pi\IZ^N\cap\Lambda_\ell^\perp)=\bigsqcup_{J\subset \{1,\ldots,N\}}\mathcal{Z}_{J},$$
where
$$\mathcal{Z}_{J}= (Z_J \cap \L_\ell^\perp) /(2\pi\IZ^N\cap\Lambda_\ell^\perp).$$

Taking the convention that  $\dim\Zc_J=-\infty$ when $\Zc_J$ is empty, we have results analogous to Lemma \ref{lem-Z-manifold}:

\begin{lemma} \label{lem-Zc-manifold}
\begin{enumerate}[\rm (i)]
\item $\Zc_\emptyset$ is a submanifold of $\T_\ell$ of dimension $N_\ell-1$ to which $\ell$ is transverse.
\item Let $J \subset \NN$ with $\Card J \geq 2$. Then $\Zc_J$ is a submanifold of $\T_\ell$ of dimension not greater than $N_\ell -1$ and 
\[
\partial \Zc_J:=\overline{\Zc_J}\setminus \Zc_J \subset \bigcup_{J'} \Zc_{J'},
\]
where the union is over the sets $J' \subset \NN$ such that $\dim(\Zc_{J'}) < N_\ell-1$.
Moreover, $\ell$ is transverse to $\Zc_J$.
\end{enumerate}
\end{lemma}

\begin{proof}
Since $\ell$ is transverse to $Z_\emptyset = \Psi^{-1}(\set 0)$ and belongs to $\L_\ell^\perp$, $Z_\emptyset \cap \L_\ell^\perp$ is a submanifold of $\L_\ell^\perp$ of dimension $N_\ell-1$ to which $\ell$ is transverse. After taking the quotient, the analogous properties hold for $\Zc_\emptyset$.

Now let $J \subset \NN$ with $\Card J \geq 2$. Since $Z_J$ coincides around each of its points with an affine subspace of $\R^N$ (with constant tangent space), the same holds for $Z_J \cap \L_\ell^\perp$. Since $\ell$ belongs to $\L_\ell^\perp$ and not to $T Z_J$, the dimension of $Z_J \cap \L_\ell^\perp$ is not greater than $N_\ell-1$ and $\ell$ is a transverse vector. All this then holds for $\Zc_J$.   

We have $\partial (Z_J\cap \L_\ell^\perp) \subset \partial Z_J \cap \L_\ell^\perp$. Let $y_0 \in \partial (Z_J\cap \L_\ell^\perp)$ and $J' \subset \NN$ such that $y_0 \in Z_{J'}\cap \L_\ell^\perp$ and $J \subsetneqq J'$. Let $\Vc$ be a neighborhood of $y_0$ in $Z_{J'}\cap \L_\ell^\perp$ such that $\Vc$ is an open subset of an affine subspace of $\R^N$. If we denote by $(f_1,\dots,f_N)$ the canonical basis of $\R^N$, there exists $j \in J' \setminus J$ such that $y_0 + tf_j$ belongs to $Z_J\cap \L_\ell^\perp$ for $t > 0$ small enough. Then for all $y \in \Vc$ and $t > 0$ small enough we have $y + t f_j \in Z_J\cap \L_\ell^\perp$. This implies that 
\[
\dim(Z_{J'}\cap \L_\ell^\perp) < \dim(Z_J\cap \L_\ell^\perp) \leq N_\ell - 1.
\]
The same conclusion holds for $\Zc_J$.
\end{proof}

%
%
%
%
%

\subsection{Distribution of the eigenvalues by ergodic averaging} \label{sec-ergodic}

We now study the distribution of the eigenvalues of $H_0$ following the seminal work of Barra and Gaspard~\cite{barrag00}. The results in this paragraph were already presented under some genericity assumptions on $\ell$ for quantum star graphs~\cite[\S 3]{KeatingMarWin03} and for general quantum graphs in~\cite[\S 4.3]{BerkolaikoWi10} and~\cite[\S 3]{cdverdiere15}. We briefly recall these arguments to see that the case of a general metric $\ell$ can be handled similarly for quantum star graphs using the conventions of the above paragraphs. For later purpose, note also that, compared with~\cite[\S 3]{KeatingMarWin03}, we will also allow test functions which are not necessarly compactly supported in $\mathcal{Z}_{\emptyset}$. To that aim, we set 
\[
\Jc_\ell = \set {J \subset \NN \st \dim(\Zc_J) = N_\ell - 1}.
\]
For $J \in \Jc_\ell$ we denote by $\m_J$ the Lebesgue measure on $\Zc_J$ and by $\nu$ a unit vector normal to $\Zc_J$ in $\T_\ell$.

\begin{lemma}\label{l:BG} Let $g$ be a continuous function on $\Zc$. We assume that $g$ is compactly supported on $\mathcal{Z}_{J}$ for some $J \in \Jc_\ell$. Then, one has
$$
\frac 1 n \sum_{k=1}^n (g\circ\varphi_\ell^{\tau_k})(0)
\limt n {+\infty} \frac{\pi m_J}{|\Gamma| \abs{\IT_\ell}}\int_{\mathcal{Z}_{J}} g(y) |\ell \cdot \nu(y)| \, d \mu_J(y).$$
\end{lemma}

In the case of a generic $\ell$, see~\cite[Th.~8]{KeatingMarWin03} for star graphs and~\cite[Prop.~4.4]{BerkolaikoWi10} and~\cite[Lemma~3.1]{cdverdiere15} for the case of more general graphs.

\begin{proof} Let $\Vc$ be a neighborhood of the support of $g$ in $\Zc_J$ such that $\overline \Vc \subset \Zc_J$. Since $\ell$ is transverse to $\Zc_J$, there exists $\delta \in ]0,\tau_1[$ such that the map 
\[
\Tc : \fonc{]-2\d,2\d[ \times \Vc} {\T_\ell} {(t,y)}{y + t \ell}
\]
is injective with an image which does not intersect $\Zc \setminus \Zc_J$. Let $\h \in {C_0^\infty(]-\d,\d[,[0,1])}$ be equal to $1$ near $0$ and such that $\int_{\IR}\chi(t)dt=\delta$. We define a continuous function $\tilde g$ on $\T_\ell$ by setting 
\[
\tilde{g}(y+t\ell)=\chi(t)g(y)
\]
for all $t \in ]-\d,\d[$ and $y \in \Vc$, and $\tilde g = 0$ outside the image of $\Tc$. Then
$$
\frac 1 n \sum_{k=1}^n (g\circ\varphi^{\tau_k}_\ell)(0) = \frac{\tau_n+\delta}{n} \frac 1 {\tau_n+\delta} \frac{m_J}{\d}\int_0^{\tau_n+\delta}\tilde{g}\circ\varphi_\ell^t(0)dt.
$$
By \eqref{eq-Weyl-H0-2} and unique ergodicity~\eqref{e:unique-ergodicity} of the flow $\varphi_\ell^t$ on $\T_\ell$, one finds that
$$
\frac 1 n \sum_{k=1}^n (g\circ\varphi^{\tau_k}_\ell)(0) \limt n {+\infty}
\frac{\pi m_J}{|\Gamma| \abs{\IT_\ell}}\frac{1}{\d}\int_{\IT_\ell}\tilde{g}
 \, d\m_{\T_\ell}.$$
Since the Jacobian of $\Tc$ at $(t,y)$ is $|\ell \cdot \nu(y)|$, one gets
$$ \frac{1}{\d}\int_{\IT_\ell}\tilde{g} \, d\m_{\T_\ell}
= \int_{\mathcal{Z}_{J}}|\ell \cdot \nu(y)|g(y) \, d\m_J(y),$$
and the conclusion follows.
\end{proof}

For $\varepsilon>0$ we introduce
\begin{equation} \label{def-Z-eps}
\mathcal{Z}(\varepsilon)=\left\{y\in\mathcal{Z}: \exists J\ \text{s.t.}\ \text{dim}(\mathcal{Z}_{J})<N_\ell-1\ 
\text{and}\ d(y,\mathcal{Z}_{J})\leq\varepsilon\right\}.
\end{equation}
In other words, we consider an $\varepsilon$-neighborhood of the strata of $\mathcal{Z}$ whose dimensions are strictly smaller than $N_\ell-1$. 
In particular, the $(N_\ell-1)$-dimensional Hausdorff measure of this set is $\mathcal{O}(\varepsilon)$.

\begin{lemma} \label{lem-Z-eps}
We have
$$\limsup_{n\rightarrow+\infty}\frac{\Card{\set{k \leq n : \varphi_\ell^{\tau_k}(0)\in\mathcal{Z}_\ell(\varepsilon)}}}{n}
\mathop{=}_{\e \to 0} \mathcal{O}(\varepsilon).$$
\end{lemma}

See~\cite[p.~354]{cdverdiere15} for the case of a generic metric on a general quantum graph.

\begin{proof}
Let
$$B(\varepsilon)=\left\{y+t\ell: -\varepsilon\leq t\leq\varepsilon\ \text{and}\ y\in\mathcal{Z}(\varepsilon)\right\}.$$
Since this set is contained in an $\mathcal{O}(\epsilon)$ neighborhood of the strata of $\mathcal{Z}$ having dimensions not greater than $N_\ell-2$, its Lebesgue measure is of size $\mathcal{O}(\varepsilon^2)$ as $\varepsilon$ tends to $0$.

Let $n \in \N^*$ and $k \leq n$ such that $\varphi_\ell^{\tau_k}(0)\in\mathcal{Z}(\varepsilon)$. Then $\varphi_\ell^{t}(0)$ stays inside $B(\varepsilon)$ for $t \in [\tau_n-\varepsilon,\tau_n+\varepsilon]$. Moreover, by the Weyl Law, the number (counted with multiplicities) of square roots of eigenvalues of $H_0$ lying in an 
interval of length $1$ is uniformly bounded by some constant $M$. 
Hence, for $\varepsilon>0$ small enough, one has
$$\frac{\Card{\set{k \leq n : \varphi_\ell^{\tau_k}(0)\in\mathcal{Z}_\ell(\varepsilon)}}}{n}
\leq\frac{M}{2\varepsilon}\int_0^{\tau_n+\varepsilon}(\mathbf{1}_{B(\varepsilon)}\circ\varphi^t_\ell)(0)dt.$$
Hence, by letting $n$ tend to $+\infty$ and by combining~\eqref{e:unique-ergodicity} with~\cite[p.149, Rk. 3]{Walters}, one finds that
$$\limsup_{n\rightarrow+\infty}\frac{\Card{\set{k \leq n : \varphi_\ell^{\tau_k}(0)\in\mathcal{Z}_\ell(\varepsilon)}}}{n}
\leq\frac{M}{2\varepsilon} \frac {\m_{\T_\ell}(B(\varepsilon))}{\abs{\T_\ell}},$$
and the conclusion follows.
\end{proof}

We define the \emph{Barra-Gaspard measure} on $\mathcal{Z}$ by
\begin{equation}\label{e:barra-gaspard}
\mu_{\BG} = \sum_{J \in \Jc_\ell} \mu_{\BG,J},
\end{equation}
where, for $J \in \Jc_\ell$, we have set 
\[
\mu_{\BG,J} =  m_J |\ell \cdot \nu(y)| \mu_{J}.
\]

Finally, we obtain the following averaging property:

\begin{proposition} \label{prop-BG}
Let $g$ be a continuous function on $\Zc$. Then, one has
$$
\frac 1 n \sum_{k=1}^n (g\circ\varphi_\ell^{\tau_k})(0)
\limt n {+\infty} \frac{\pi}{|\Gamma| \abs{\IT_\ell}}\int_{\mathcal{Z}} g \, d \mu_{\BG}.$$
\end{proposition}

Again, we refer to~\cite{KeatingMarWin03, BerkolaikoWi10, cdverdiere15} for earlier versions of this result for quantum graphs endowed with Kirchhoff conditions and with a generic metric $\ell$. In particular, this proposition implies
\begin{equation} \label{eq-measure-BG}
\m_{\BG}(\Zc) = \frac{|\Gamma| \abs{\IT_\ell}}{\pi}.
\end{equation}

\begin{proof}
Let $\e > 0$. We introduce a family of nonnegative 
continuous functions $\tilde \chi$ and $(\chi_J)_{J \in \Jc_\ell}$ on $\mathcal{Z}$ such that 
$$\tilde \chi+\sum_{J \in \Jc_\ell}\chi_J=1,$$ 
$\tilde \chi$ is supported in $\mathcal{Z}(\varepsilon)$ and $\chi_J$ is compactly supported in $\mathcal{Z}_{J}$. By Lemmas \ref{l:BG} and \ref{lem-Z-eps} applied to $\chi_J g$, $J \in \Jc_\ell$, and $\tilde \chi g$, respectively, we have 
\[
\lim_{n \to +\infty}\frac 1 n \sum_{k=1}^n (g\circ\varphi_\ell^{\tau_k})(0) = \sum_{J \in \Jc_\ell} \int_{\Zc_J} \h_J g \, d\m_{\BG,J} + \Oc(\e).
\]
We let $\e$ go to 0, and by the dominated convergence Theorem we obtain the expected result.
\end{proof}

\section{Separation and localization of the eigenvalues} \label{sec-separation}

In this section, we will give an acccurate description of the localization of the eigenvalues of $H_\a$ when $\a\in\C$. This will rely on a careful analysis that will allow us to view in a quantitative manner the eigenvalues of $H_\a$ as a perturbation of the eigenvalues of $H_0$ as the real part of the spectral parameter tends to $+\infty$. We proceed in two steps. First, we localize a sequence of eigenvalues of $H_\a$ near the eigenvalues of $H_0$. Then, we show that we have indeed found all the eigenvalues. All this analysis is summarized in Proposition~\ref{prop-zna} which is the main result of this section and from which Propositions~\ref{prop-spectrum-Ha},~\ref{th-Weyl},~\ref{prop-eigenvalue-proche} and~\ref{prop-eigenvalue-Dir-alpha} follow.

To that aim, given $\a \in \C$ and $z \in \C^* \setminus \Tc_\infty$, we set 
\[
\psi_\a (z) = \Psi(z\ell) - \frac \a z = - \sum_{j=1}^N \cotan(z \ell_j) - \frac \a z.
\]
The functions $\psi_\a$, $\a \in \C$, are holomorphic on $\C^* \setminus \Tc_\infty$ and $z^2$ is an eigenvalue of $H_\a$ inside 
$\C^* \setminus \Tc_\infty$ if and only if $\psi_\a(z) = 0$ (see the discussion of Paragraph \ref{sec-secular}). In particular, the positive roots of $\psi_0$ 
outside $\Tc_\infty$ are given by the $\tau_n$ such that $\tau_n^2$ is not an 
eigenvalue of $H_{\infty}$. 

\subsection{Upper bound on the set of eigenvalues}

We set 
\[
\Nc = \set{n \in \N^* \st \tau_n \notin \Tc_\infty},
\]
and, for $n \in \Nc$,
\[
\rho_n = \frac 1 {\psi_0'(\tau_n)}=\frac{1}{\nabla_{\tau_n\ell}\Psi \cdot \ell}.
\]
Then, the following holds:

\begin{lemma} \label{lem-g}
\begin{enumerate}[\rm (i)]
\item For $n \in \Nc$, one has 
\begin{equation} \label{estim-rau-n}
0 < \rho_n \leq \min \left\{ \abs \G\inv , \frac{\sin(\tau_n \ell_1)^2}{\ell_1}, \dots, \frac {\sin(\tau_n \ell_N)^2}{\ell_N} \right\}.
\end{equation}
%
\item Let $\g_0 = \frac 1 {16eN}$.
Then, for $n \in \Nc$, one has 
$D(\tau_n,\g_0 \rho_n) \cap \Tc_\infty = \emptyset$ and, for $\g \in [0,\gamma_0]$ and $z \in D(\tau_n,\g \rho_n)$,
\begin{equation}\label{eq-g}
\abs{\psi_0(z) - \frac {z-\tau_n}{\rho_n}} \leq \frac \g 2.
\end{equation}
\end{enumerate} 
\end{lemma}

\begin{proof}
\stepp For $z \in \C^* \setminus \Tc_\infty$, one has
\begin{equation} \label{expr-g'}
\psi_0'(z) = \sum_{j=1}^N \frac {\ell_j}{\sin(\ell_j z)^2} = \sum_{j=1}^N \ell_j \big( 1 + \cotan^2 (z\ell_j) \big),
\end{equation}
and
\begin{equation} \label{expr-g''}
\psi_0''(z) = - 2 \sum_{j=1}^N \frac {\ell_j^2 \cos(z\ell_j)}{\sin(z\ell_j)^3}.
\end{equation}
In particular, \eqref{expr-g'} gives \eqref{estim-rau-n}.

\stepp Let $\gamma \in ]0,\gamma_0]$, $n \in \Nc$, $j \in \NN$ and $z \in D(\tau_n,\g \rho_n)$. Since $\gamma \leq 1$ and $\ell_j \rho_n \leq 1$ one has 
\[
\abs {\cos(z \ell_j)} \leq e^{\ell_j \g \rho_n} \leq e,
\]
and hence, since we also have $\g \leq 1 / 2e$,
\begin{align*}
\abs{\sin(z \ell_j)}
 \geq \abs{\sin(\tau_n \ell_j)} - \abs{z-\tau_n} \ell_j e
 \geq \ell_j^{\frac 1 2} \rho_n^{\frac 12}  - \gamma \rho_n \ell_j e
 \geq \frac {\ell_j^{\frac 12} \rho_n^{\frac 12}}2. 
\end{align*}
In particular, $F_D(z\ell) \neq 0$, so $D(\tau_n,\g \rho_n) \cap \Tc_\infty=\emptyset$. 
Then, by \eqref{expr-g''} we have 
\begin{equation}\label{e:bound-second-derivative}
\abs{\psi_0''(z)} \leq \frac {16 e}{\rho_n^{3/2}} \sum_{j=1}^N \sqrt {\ell_j}.
\end{equation}
From this one can infer that, for every $(z,z_1)$ in $D(\tau_n,\gamma \rho_n)$,
\[
\frac {\abs {z-\tau_n}^2}{2} \abs{\psi_0''(z_1)} \leq 8 \gamma^2 e \sum_{j=1}^N \sqrt {\rho_n \ell_j} \leq \frac \g 2,
\]
and the result follows by Taylor expansion.
\end{proof}

\begin{center}
\begin{figure}[h]
\includegraphics[width = 0.8 \linewidth]{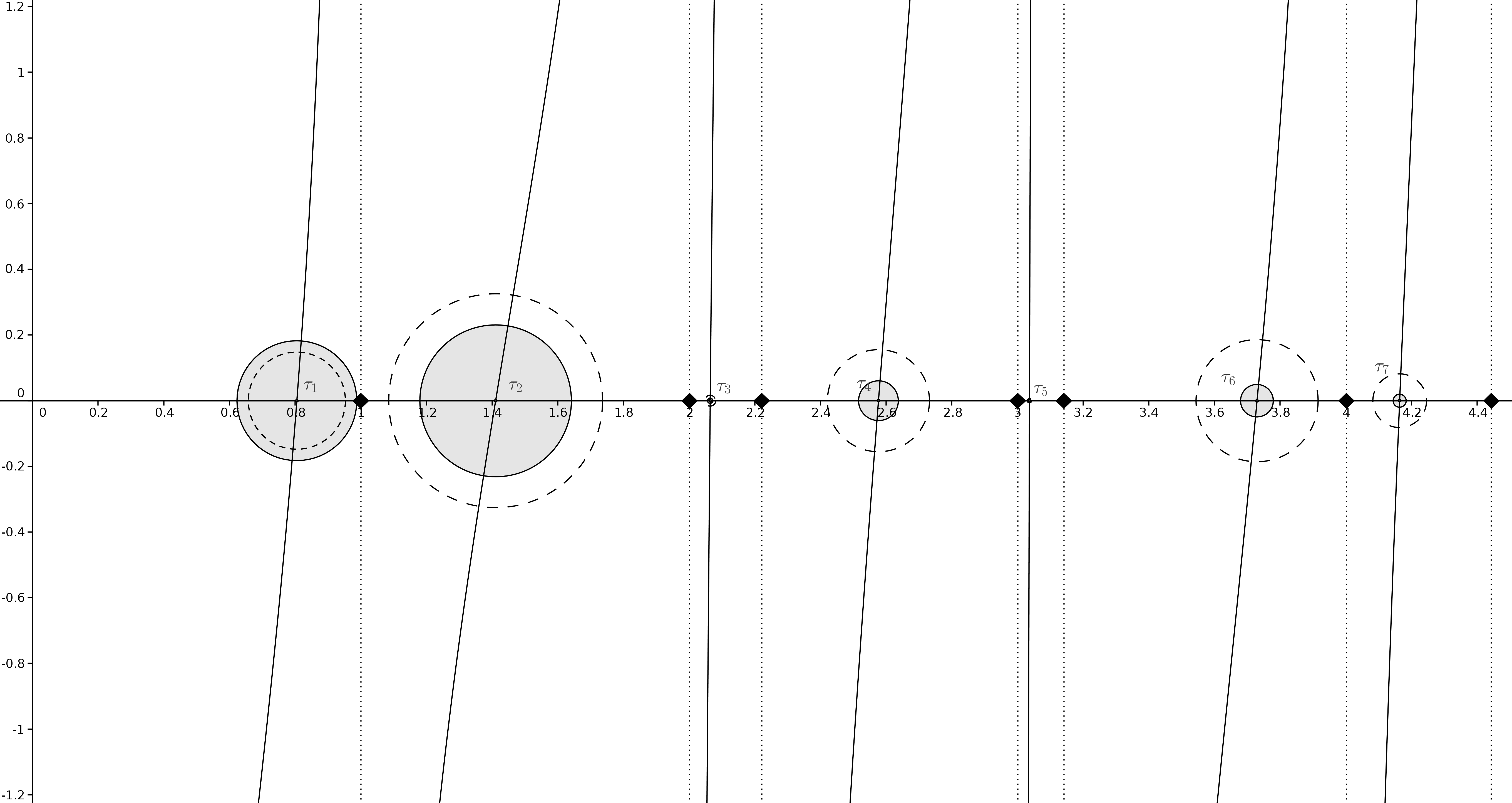}
\caption{The graph of $\Psi$ for $\ell = (1,\sqrt 2,\pi)$, with the elements of $\Tc_\infty$ (diamonds), $\tau_1$ to $\tau_7$, and the corresponding disks of radii $2 \rho_n / \tau_n$ (filled) and $2 \rho_n$ (dotted).}
\label{fig-Psi}
\end{figure}
\end{center}

\begin{proposition} \label{prop-sep-vp}
Let $\gamma_0 > 0$ be given by Lemma \ref{lem-g}. Let $r > 0$. There exists $n_r \in \N^*$ such that, for every 
$n \in \Nc$ with $n \geq n_r$ and for every $\a \in D(r)$, there is a unique $z_n(\a) \in D(\tau_n, \gamma_0 \rho_n)$ 
such that $z_n(\a)^2$ is an eigenvalue of $H_\a$. Moreover, $z_n(\a)$ depends analytically on $\a\in D(r)$, 
$z_n(\a)^2$ is an algebraically simple eigenvalue of $H_\a$, and 
\begin{equation} \label{estim-rn-alpha}
\abs{z_n(\a) - \tau_n} \leq \frac {8r\rho_n}{\tau_n}.
\end{equation}
\end{proposition}

\begin{proof}
For $\g \in ]0,\gamma_0]$, $n \in \Nc$ and $z$ on the circle $C(\tau_n,\g \rho_n)$, one has by \eqref{eq-g}
\[
\abs{\psi_0(z)} \geq \frac \g 2.
\]
On the other hand, since $\rho_n\leq|\Gamma|^{-1}$, one has, for $n$ large enough,
\[
\frac {\abs \a}{\abs z} \leq \frac r {\tau_n - \g \rho_n} \leq \frac {2r} {\tau_n}.
\]
Thus, there exists $n_r \in \N^*$ such that for $n \in \Nc$ with $n \geq n_r$, $\g \in \big[ \frac {8r}{\tau_n},\gamma_0 \big]$ and $z \in C(\tau_n,\g \rho_n)$ we have 
\[
\abs{\psi_\a(z)} \geq \abs{\psi_0(z)} - \frac {\abs \a}{\abs z}\geq  \frac {2r} {\tau_n} > 0.
\]
This proves that, for every $\a \in D(r)$ and every $\g \in \big[ \frac {8r}{\tau_n},\gamma_0 \big]$, 
the operator $H_\a$ has no eigenvalue on the set $\set {z^2 , z \in C(\tau_n ,\g \rho_n)}$. 
By continuity of the spectrum, the number of eigenvalues counted with multiplicity in 
$\set {z^2 , z \in D(\tau_n ,\g \rho_n)}$ does not depend on $\a \in D(r)$. 

By Remark \ref{rem-Weyl-H0}, $\tau_n$ is the only $\tau$ in $D(\tau_n,\g \rho_n)$ such that $\tau^2$ is an eigenvalue of $H_0$. Thus, for $\a \in D(r)$, there is exactly one $z_n(\a) \in D(\tau_n,\g\rho_n)$ such that $z_n(\a)^2$ is an eigenvalue of $H_\a$. The analyticity of this eigenvalue comes from the usual perturbation theory for parameter dependant operators (see \cite{kato}).
\end{proof}

From this point on, we set, for $r>0$, 
\[
\Nc_r = \set{n \in \Nc : n \geq n_r}.
\]

\subsection{Lower bound on the set of eigenvalues}

At this stage, \emph{we have identified some eigenvalues of $H_\a$}. If $\tau^2$ is an eigenvalue of $H_0$ which is also an eigenvalue of $H_\infty$ then it is also an eigenvalue of $H_\a$, and the multiplicities of $\tau^2$ as an eigenvalue of $H_0$ or $H_\a$ coincide. Otherwise, $\tau^2$ is a simple eigenvalue of $H_0$, and around $\tau^2$ (in a sense made precise by Proposition \ref{prop-sep-vp}), there is a unique simple eigenvalue of $H_\a$. \emph{Our purpose is now to prove that we have in fact found all the eigenvalues of $H_\a$}. For this we will need the following lemma.

\begin{lemma} \label{lem-min-rho}
We have 
\[
\limsup_{\substack{n\in\Nc \\ n \to +\infty}} \rho_n > 0.
\]
\end{lemma}

\begin{proof}
Let $\chi$ be a continuous function from $\Zc$ to $[0,1]$ compactly supported in $\Zc_\emptyset$ and equal to 1 on some open subset thereof. Then the function 
\[
f : y \mapsto \frac {\h(y)}{\nabla_y \Psi \cdot \ell}
\]
can be seen as a continuous function on $\Zc$, compactly supported in $\Zc_\emptyset$. By Lemma \ref{l:BG} we have 
\[
\frac 1 n \sum_{k=1}^n (f \circ \f_\ell^{\tau_k})(0) \limt n {+\infty} \frac {\pi}{\abs \G \abs{\T_\ell}} \int_{\Zc_\emptyset} f \, d\m_{\BG,\emptyset} > 0.
\]
This implies in particular that $\rho_n$ does not go to 0 as $n \in \Nc$ goes to $+\infty$.
\end{proof}

Now we can state the main result of this section which compares the set of eigenvalues of $H_\a$ with the ones of $H_0$.

\begin{proposition} \label{prop-zna}
Let $r > 0$ and $n_r \in \N^*$ be given by Proposition \ref{prop-sep-vp}. Then, for every $\a \in D(r)$, the spectrum of $H_\a$ is given by a sequence of eigenvalues $(z_n(\a)^2)_{n \in \N^*}$ (repeated according to algebraic multiplicities) such that the following assertions hold for $n \geq n_r$.
\begin{enumerate}[\rm(i)]
 \item If $n \notin \Nc$ then $z_n(\a) = \tau_n$.
 \item If $n \in \Nc$ then $z_n(\a)$ is given by Proposition \ref{prop-sep-vp}. In particular, there exists a disk in $\C$ centered at $\tau_n$ which contains $z_n(\a)$ but does not contain $\tau_m$ or $z_m(\a)$ for any $m \neq n$.
 \item The functions $\a\mapsto z_n(\a)$ are analytic,
 \item The geometric and algebraic multiplicity of the eigenvalue $z_n(\a)^2$ are equal, and in particular this multiplicity is 1 if $n \in \Nc$.
\end{enumerate}
\end{proposition}

Propositions~\ref{prop-spectrum-Ha},~\ref{th-Weyl},~\ref{prop-eigenvalue-proche} and~\ref{prop-eigenvalue-Dir-alpha} are consequences of the previous proposition combined with Remark~\ref{rem-Weyl-H0} (for the proof of Proposition~\ref{th-Weyl}).

\begin{center}
\begin{figure}[h]
\includegraphics[width = 0.7 \linewidth]{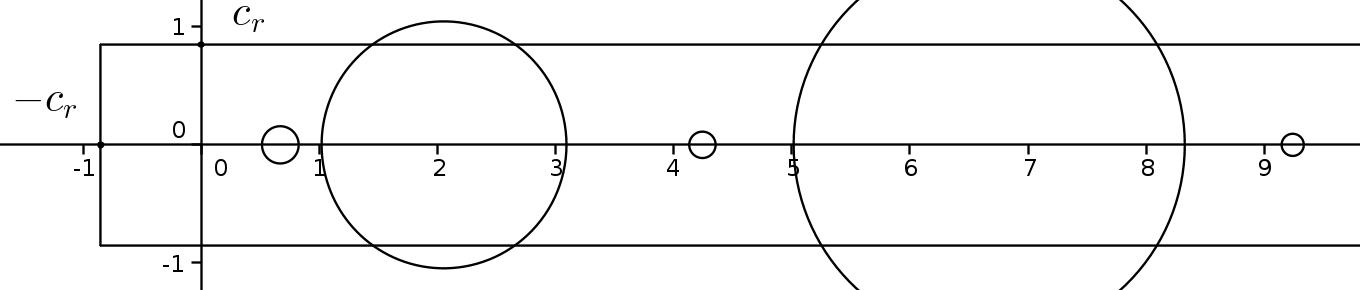}
\caption{The sets $\set{z^2 \, : \, \abs{z - \tau_n} = \g_0 \rho_n}$ for $\ell = (1,\sqrt 2,\pi)$. They are close to the disks centered at $\tau_n^2$ with radius $2\gamma_0 \tau_n \rho_n$. By Lemma \ref{lem-min-rho}, this radius is regularly greater than $c_r$.}
\label{fig-separation}
\end{figure}
\end{center}

\begin{proof}
Let $c_r$ be given by Proposition \ref{prop-strip}. We set 
\[
\O_r = \set{\l \in \C : \Re(\l) \geq -c_r, \abs{\Im(\l)} \leq c_r} \setminus  \bigcup_{n \in \Nc_r} \set{z^2 : \frac {8r \rho_n}{\tau_n} < \abs{z -\tau_n} < \gamma_0 \rho_n} .
\]
By Propositions \ref{prop-strip} and \ref{prop-sep-vp}, we have $\Sp(H_\a) \subset \O_r$ for all $\a \in D(r)$. By Lemma \ref{lem-min-rho}, all the connected components of $\O_r$ are bounded (see Figure \ref{fig-separation}). By the perturbation theory of Kato \cite{kato}, the number of eigenvalues of $H_\a$ (counted with algebraic multiplicities) in each of these connected components does not depend on $\a$. Choosing $n_r$ larger if necessary, we can assume that there exists a finite union $\tilde \O_r$ of connected components of $\O_r$ which contains $\tau_n^2$ if and only if $n < n_r$. Then, for all $\a \in D(r)$, $\tilde \O_r$ contains exactly $n_r-1$ eigenvalues of $H_\a$ (counted with multiplicities), which can be denoted by $z_n(\a)^2$, $1 \leq n < n_r$. The other connected components of $\O_r$ contain eigenvalues which have already been identified. If such a connected component contains $\tau_n^2$ for some $n \in \Nc$, $n \geq n_r$, it contains exactly one eigenvalue of $H_\a$, which we have denoted by $z_n(\a)^2$. We also know that for $n \notin \Nc$, the eigenvalue $\tau_n^2$ remains an eigenvalue of $H_\a$ with the same multiplicity. Since the number of eigenvalue in each connected components is fixed, there cannot be any other eigenvalue for $H_\a$ in $\O_r$. All the properties enumerated in the proposition then follow from Lemma \ref{l:multiplicity} and Proposition \ref{prop-sep-vp}.
\end{proof}

In this proof, it is important that all the connected components of $\O_r$ are bounded. Otherwise, even with the continuity of the spectrum, an eigenvalue could disappear (by escaping to infinity) or be created (by coming from infinity).

\section{Proofs of the main results} \label{sec-proofs}

Let $\a \in \C$. We have seen that for $n$ large enough, the eigenvalue $z_n(\a)^2$ of $H_\a$ is either equal to $\tau_n^2$, or close to $\tau_n^2$ in the sense of Proposition \ref{prop-sep-vp}. Now we prove that the difference $\d_n(\a)$ (see \eqref{def-delta}) is, at first order, a simple function of $\tau_n$, and we use the results of Section \ref{sec-ergodic} to prove Theorem \ref{t:maintheo2}. Finally, as a by-product of our analysis, we will prove Propositions \ref{t:maintheo3} and \ref{t:maintheo4}.

\subsection{The limit measure}

\begin{lemma} \label{lem-delta}
Let $r > 0$ and let $n_r \in\N^*$ be given by Proposition \ref{prop-sep-vp}. Then, for $n \in \Nc_r$, we have 
\[
\d_n(\a) = 2\a \rho_n + \mathop{\Oc}_{n \to +\infty} (n\inv), 
\]
where the rest is uniform with respect to $\a \in D(r)$.
\end{lemma}

This asymptotic expression will allow us to relate the distribution of $\delta_n(\alpha)$ for large $n$ to the Barra-Gaspard measure, as in~\cite[Th.~3]{KeatingMarWin03} where the asymptotic distribution of $\psi_0'(\tau_n)=\rho_n^{-1}$ on compact intervals of $\R$ was described for $\Lambda_\ell=\{0\}$ (in that case, $\mathcal{N}=\N^*$).

\begin{proof}
For $\a \in D(r)$ and $n \in \Nc_r$, we set 
\[
\eta_n(\a) = z_n(a) - \tau_n.
\]
By \eqref{estim-rn-alpha} and \eqref{estim-rau-n}, we have 
\[
\eta_n(\a) = \Oc(\rho_n \tau_n^{-1}) = \Oc(\tau_n^{-1}).
\]
Then, by \eqref{e:bound-second-derivative} and \eqref{estim-rn-alpha},
$$
\abs{\psi_0(z_n(\alpha)) - \eta_n(\a) \psi_0'(\tau_n)} \leq \frac {\abs{\eta_n(\a)}^2}2 \sup_{\abs{z-\tau_n} \leq \abs{\eta_n(\a)}} \abs{\psi_0''(z)} \lesssim \frac {\rho_n^{1/2}}{\tau_n^2} \lesssim \frac 1 {\tau_n^2}.
$$
We have
$$
0 = \psi_0 \big( \tau_n+ \eta_n(\a) \big) - \frac{\alpha}{\tau_n+ \eta_n(\a)} = \eta_n(\a) \psi_0'(\tau_n) - \frac \a {\tau_n} + \Oc(\tau_n^{-2}).
$$
Since $\rho_n$ is bounded, this gives 
\begin{equation} \label{eq-eta-n}
\eta_n(\a) = \frac {\a \rho_n}{\tau_n} + \Oc(\tau_n^{-2}),
\end{equation}
and then 
\begin{equation} \label{eq-delta-n}
\d_n(\a) = 2 \tau_n \eta_n(\a) + \eta_n(\a)^2 = 2\a \rho_n + \Oc(\tau_n^{-1}).
\end{equation}
We conclude with \eqref{eq-Weyl-H0-2}.
\end{proof}

We consider on $\Zc$ the continuous function defined by
\[
\Phi (y) = \frac {2} {\nabla_y \Psi \cdot \ell} = \frac{2}{\sum_{j=1}^N\ell_j(1+\cotan^2(y_j))}, \qquad  \text{if} \quad y \in \Zc_\emptyset,
\]
and $\Phi(y) = 0$ otherwise.

\begin{proposition}
The first statement of Theorem \ref{t:maintheo2} holds if we set 
\begin{equation} \label{def-mu-ell}
\mu_\ell =  \frac \pi {\abs \G \abs{\T_\ell}}  \left(\Phi_* \m_{\BG,\emptyset}  + \d_0 \sum_{J \in \Jc_\ell, J \neq \emptyset} \m_{\BG,J}(\Zc_J) \right),
\end{equation}
where $\Phi_* \m_{\BG,\emptyset}$ is the pushforward of the measure $\m_{\BG,\emptyset}$ by $\Phi$.
\end{proposition}

\begin{proof}
We first observe that the image of $\Phi$ is included in $[0,2\abs \G\inv]$ so, by \eqref{eq-measure-BG}, the right-hand side of \eqref{def-mu-ell} defines a probability measure on $[0,2\abs \G\inv]$.

Let $\a \in \C$. Let $f$ be a continuous function on $\C$. By Lemma \ref{lem-delta} we have 
\[
\frac 1 n \sum_{k=1}^n f(\d_n(\a)) = \frac 1 n \sum_{k=1}^n f(\a \Phi(\f_\ell^{\tau_n}(0))) + \mathop{\Oc}_{n \to +\infty} (n\inv).
\]
Then, by Proposition \ref{prop-BG} applied to $g = f \circ (\a \Phi)$, 
\[
\frac 1 n \sum_{k=1}^n f(\d_n(\a)) \limt n {+\infty} \frac \pi {\abs \G \abs{\T_\ell}} \left( \int_{\Zc_\emptyset} (f \circ (\a \Phi)) \, d\m_{\BG,\emptyset} + f(0) \sum_{J \in \Jc_\ell, J \neq \emptyset} \m_{\BG,J}(\Zc_J) \right).
\]
The proposition follows.
\end{proof}

We observe that if $\ell_j / \ell_1 \in \Q$ for all $j \in \NN$ then $\L_\ell^\perp$ is of dimension 1. So, by Lemma \ref{lem-Zc-manifold}, $\Zc$ consists of isolated points, and $\Phi$ only takes a finite number of values (including 0, at the point 0). This implies that $\mu_\ell$ is a finite sum of Dirac masses, including the Dirac mass at 0.

To conclude the proof of Theorem \ref{t:maintheo2}, it remains to prove the last statement. It is given by the following proposition.

\begin{proposition}\label{prop:support-measure} Suppose that $\Lambda_\ell=\{0\}$. Then the measure $\mu_\ell$ is absolutely continuous with respect to the Lebesgue measure and its support is the interval $[0,2|\Gamma|^{-1}]$.
\end{proposition}

\begin{proof} 
We first observe that when $N_\ell = N$ we have $\Jc_\ell = \set \emptyset$. Then, for $f \in C_c^0(\R)$, we have 
$$
\int_\R f \, d\mu_\ell = \frac{\pi}{|\Gamma|(2\pi)^N}\int_{\mathcal{Z}_{\emptyset}} (f \circ \Phi) \, d\mu_{\BG,\emptyset}.
$$

The image of $\Phi$ is included in $[0,2\abs \G\inv]$, and considering the point $\big(\th,-\th, \frac \pi 2, \dots, \frac \pi 2\big)$ (modulo $2\pi$) for $\th$ going from $\frac \pi 2$ to $0$, we see that the image of $\Phi$ is in fact exactly $[0,2\abs \G\inv]$.

We extend $\Phi$ to a map $\tilde \Phi : y \mapsto 2 (\nabla_y \Psi \cdot \ell) \inv$ on an open neighborhood of $\Zc_\emptyset$ in $\T^N$. For $y \in \Zc_\emptyset$ we have 
$$
\nabla_y\tilde \Phi
=\Phi(y)^2 \left(\frac{\ell_1\cos(y_1)}{\sin^3(y_1)},\ldots,\frac{\ell_N\cos(y_N)}{\sin^3(y_N)}\right)
$$
and
$$
\nabla_y\Psi=\left(\frac{1}{\sin^2(y_1)},\ldots,\frac{1}{\sin^2(y_N)}\right).
$$
If these two vectors are colinear we have 
\[
\ell_1 \cotan(y_1) = \dots = \ell_N \cotan(y_N).
\]
Since $\Psi(y) = 0$, this only happens when $y_j=\frac{\pi}{2}\ \text{mod}\ \pi$ for all $j \in \NN$. Then, for $t\in (0,2|\Gamma|^{-1})$, the subset $\Phi^{-1}(\set t)$ is a smooth submanifold of codimension $1$ in $\mathcal{Z}_{\emptyset}$.

Let $f \in \mathcal{C}^\infty_c((0,2|\Gamma|^{-1}))$. Applying the co-area formula (see \cite{Federer}), we can write
$$\int_{\mathcal{Z}_{\emptyset}}
 (f\circ \Phi) \, d\mu_{\BG,\emptyset}
 =\int_0^{2|\Gamma|^{-1}}f(t)
 \left(\int_{\Phi^{-1}(\set t)} \frac{|\nu(y) \cdot \ell|}{\|\nabla \Phi \|} \, d\mu_{\Phi^{-1}(\set t)}\right)dt,
$$
where $\mu_{\Phi^{-1}(\set t)}$ is the Lebesgue measure on $\Phi\inv(\set t)$.
Hence, the measure $\mu_\ell$ is absolutely continuous with respect to the Lebesgue measure 
on the interval $(0,2|\Gamma|^{-1})$. We can also note that its density is positive almost everywhere 
since $\abs{\nu(y)\cdot\ell} > 0$ on $\mathcal{Z}_{\emptyset}$ and $\Phi\inv(\set t)$ has positive measure for all $t \in ]0,\abs \Gamma\inv[$. Since $\Phi$ does not vanish on $\Zc_\emptyset$, the measure $\mu_\ell$ does not put any mass at $t=0$, and $\Phi^{-1}(\{2|\Gamma|^{-1}\})$ is reduced to $2^N$ points so $\mu_{\BG,\emptyset}\left(\Phi^{-1}(\{|\Gamma|^{-1}\})\right)=0$.
\end{proof}

\begin{remark}
More generally, we can consider the case where $1<N_\ell \leq N$. We denote by $P_\ell$ the 
orthogonal projection of $\R^N$ on the vector space $\Lambda_\ell^\perp.$ Then, we write down the partition
$$\mathcal{Z}_{\emptyset}=\mathcal{Z}_{\emptyset}^{\text{reg}}\sqcup \mathcal{Z}_{\emptyset}^{\text{sing}},$$
where
$$\mathcal{Z}_{\emptyset}^{\text{reg}}=\{y\in\mathcal{Z}_{\emptyset}:\text{rk}(P_\ell\nabla_y\Psi,
P_\ell\nabla_y\tilde{\Phi})=2\},$$
and
$$\mathcal{Z}_{\emptyset}^{\text{sing}}=\{y\in\mathcal{Z}_{\emptyset}:\text{rk}(P_\ell\nabla_y\Psi,
P_\ell\nabla_y\tilde{\Phi})=1\}.$$
We can decompose the integral following this partition of $\mathcal{Z}_{\emptyset}$:
$$
\int_{\mathcal{Z}_{\emptyset}} (f\circ\Phi) \, d\mu_{\BG,\emptyset}=\int_{\mathcal{Z}_{\emptyset}^{\text{reg}}} (f\circ\Phi) \, d\mu_{\BG,\emptyset}+\int_{\mathcal{Z}_{\emptyset}^{\text{sing}}} (f\circ\Phi) \, d\mu_{\BG,\emptyset}.
$$
To the first term, we can again apply the co-area formula
$$
\int_{\mathcal{Z}_{\emptyset}^{\text{reg}}} (f\circ\Phi) \, d\mu_{\BG,\emptyset} = \int_0^{2|\Gamma|^{-1}} f(t)
 \left(\int_{\Phi^{-1}(\set t)\cap\mathcal{Z}_{\emptyset}^{\text{reg}}}\frac{|\nu(y) \cdot \ell|}{\|\nabla\Phi\|} \, d\m_{\Phi^{-1}(\set t)}\right)dt,$$
which gives as above a contribution to the absolutely continuous part of the the measure $\mu_\ell$. 
In particular, if $$\m_{\mathcal{Z}_{\emptyset}}(\mathcal{Z}_{\emptyset}^{\text{sing}})=0,$$ then 
the limit measure is a linear combination of a Dirac mass at $0$ and of an absolutely continuous measure.
\end{remark}

\subsection{Rates of convergence}

In this paragraph we prove Propositions \ref{t:maintheo3} and \ref{t:maintheo4}. 

\begin{proof}[Proof of Proposition \ref{t:maintheo3}]
We already know that if two lengths are commensurable then $\lambda_n(\alpha)=\lambda_n(0)$ (and hence $\d_n(\a) = 0$) for an 
infinite number of $n \in \N^*$. Now assume that, for instance, $\ell_1/\ell_2$ does not belong to $\mathbb{Q}$. 

We have seen in Proposition \ref{prop-sep-vp} that $\d_n(\a)$ is small when $\rho_n$ is small. This happens when $\tau_n$ is close to $\Tc_\infty$. By the Dirichlet Approximation Theorem, there exists a sequence $((p_k,q_k))_{k \in \N}$ in $\IZ_+^2$ such that $q_k$ goes to $+\infty$ and, for all $k \in \N$,
$$
\left|\frac{\ell_1}{\ell_2}-\frac{p_k}{q_k}\right|\leq\frac{1}{q_k^2}.
$$
Thus, one has
$$
\left|\frac{\pi q_k}{\ell_2}-\frac{\pi p_k}{\ell_1}\right|\leq\frac{\pi}{q_k\ell_1}.
$$
Since $\pi q_k/\ell_1$ and $\pi p_k / \ell_2$ belong to $\Tc_\infty$, there exists $n_k \in \N^*$ such that $\tau_{n_k}$ is between them (see Remark \ref{rem-Weyl-H0}). In particular,
\[
\abs{\tau_{n_k} \ell_2 - \pi q_k} = \ell_2 \abs{\tau_{n_k}- \frac {\pi q_k}{\ell_2}} \leq \frac {\pi \ell_2} {q_k \ell_1},
\]
so, by Lemma \ref{lem-g},
\[
\rho_{n_k} \lesssim \sin(\tau_{n_k} \ell_2)^2 \lesssim \frac 1 {q_k^2} \lesssim \frac 1 {\tau_{n_k}^2}.
\]
By Proposition~\ref{prop-sep-vp}, there exists an eigenvalue $\lambda_{n_k}(\alpha)=z_{n_k}(\alpha)^2$ of $H_\alpha$ such that 
$$
z_{n_k}(\a)\in D \left(\tau_{n_k} , \frac {8|\alpha| \rho_{n_k}}{\tau_{n_k}} \right).
$$ 
Then
\[
\abs{\d_{n_k}(\a)} \lesssim \tau_{n_k} \abs{z_{n_k}(\a) - \tau_{n_k}} \lesssim \frac 1 {\tau_{n_k}^2},
\]
and the conclusion follows.
\end{proof}

Let $\gamma>0$, we say that $\ell$ is $\gamma$-Diophantine (see \cite{Poschel01}) if there exists $C > 0$ such that, for every $\kappa\in\IZ^N\setminus\{0\}$,
\begin{equation}
\label{e:diophantine}|\kappa \cdot \ell|\geq C\| \kappa \|^{-\gamma},
\end{equation}
where $\| \cdot \|$ is the Euclidean norm on $\IR^N$. In particular, $\k \cdot \ell \neq 0$ for all $\k \in \Z \setminus \set 0$ in this case. According to~\cite{Poschel01}, the set 
of $\gamma$-Diophantine vectors has full Lebesgue measure for $\gamma>N-1$ and is empty when $\gamma<N-1$. The Diophantine properties of $\ell$ 
allow to be more precise on the unique ergodicity of the flow $\varphi^t_{\ell}$:

\begin{lemma}\label{l:quantitative-birkhoff}
 Suppose that $\ell$ is $\gamma$-Diophantine for some $\gamma>0$. Then there exists $C_1>0$ such that
 for $a\in\mathcal{C}^{\infty}(\IT^N)$, $T>0$ and $x\in\IT^N$ one has
 $$\left|\frac{1}{T}\int_0^T a\circ\varphi_\ell^t(x)dt - \frac{1}{(2\pi)^N} \int_{\IT^N} a(y) \, dy\right|
 \leq\frac{C_1}{T}\sum_{\kappa \in\IZ^N \setminus \set 0}\|\kappa \|^{\gamma}\left|\frac{1}{(2\pi)^N} \int_{\IT^N}a(y)e^{-i \kappa \cdot y}\, dy \right|.$$
\end{lemma}

\begin{proof} The proof is immediate once we have decomposed $a$ in Fourier series. 
\end{proof}

We now proceed to the proof of Proposition~\ref{t:maintheo4}.

\begin{proof} [Proof of Proposition \ref{t:maintheo4}]
Let $\gamma \in ]N-1,N[$. Let $\O\subset (\R_+^*)^N$ be of full Lebesgue measure and such that any vector in $\O$ is $\gamma$-Diophantine. We assume that $\ell$ belongs to $\O$.
Let $\sigma \in ]0,2|\Gamma|^{-1}]$ and let $y_0$ in $\Zc_\emptyset$ be such that 
\begin{equation}
\Phi(y_0) = \sigma.
\end{equation}
Let $\h \in C_0^\infty(\R^N,[0,1])$ be supported in the unit ball and be equal to $1$ on a small neighborhood of 0. For $y \in \R^N$ and $r > 0$ we set 
\[
\h_r(y) = \sum_{\k \in \Z^N} \h \left( \frac {y-y_0 + 2\k \pi}{r} \right).
\]
Here we have identified $y_0$ with a point in $\R^N$, and then $\h_r$ can be seen as a function on $\T^N$. Let $\k \in \Z^N \setminus \set 0$ and $j \in \NN$ such that $\abs{\k_j} = \max_{1\leq k \leq N} \abs{\k_k}$. Then, for $p \in \N$, we have 
\begin{equation} \label{estim-reste-IPP}
\abs{\int_{\T^N} \h_r(y) e^{-i\k\cdot y} \frac {dy}{(2\pi)^N}} = \frac 1 {\k_j^p} \abs{\int_{\T^N} \partial_{y_j}^p  \h_r(y)  e^{-i\k\cdot y} \frac {dy}{(2\pi)^N}} \leq C r^{N-p} \nr{\k}^{-p},
\end{equation}
where $C$ only depends on $\h$ and $p$. Let $T > 0$ and $r \in (0,1]$. We apply Lemma \ref{l:quantitative-birkhoff} with $a = \h_r$ and $x = \f_\ell^T(0)$, and \eqref{estim-reste-IPP} with $p = 2N$. This gives 
$$\left|\int_{T}^{2T} (\chi_{r}\circ\varphi_\ell^\tau) (0) \, d\tau - r^N T \int_{\IT^N}\chi_{1}(y)\frac{dy}{(2\pi)^N}\right|
 \leq Cr^{-N}.$$
Let $\e > 0$. For $T \geq 1$ we set 
 \[
 r_T = T^{-\frac 1 {2N} + 2\e}.
 \]
Then
$$
\int_{T}^{2T} (\chi_{r_T}\circ\varphi_\ell^\tau)(0) \, d\tau \mathop{\sim}_{T \to +\infty}  r_T^N T\int_{\IT^N}\chi_{1}(y)\frac{dy}{(2\pi)^N}.
$$
In particular, for $T$ large enough the left-hand side is positive, so there exists $\tau \in [T,2T]$ such that $\h_{r_T}(\f_\ell^\tau(0)) \neq 0$, and in particular $\abs{\f_\ell^\tau(0) - y_0} \leq r_T$. Now, as $\ell$ is transverse to $\Zc_\emptyset$, there exists $c > 0$ independant of $T$ large enough and $n_T \in \N^*$ such that $\tau_{n_T} \in \big[\frac T 2, 3T\big]$ and $\abs{\f_\ell^{\tau_{n_T}}(0) - y_0} \leq c r_T$. By \eqref{eq-delta-n} and as $\Phi$ is of class $\mathcal{C}^1$ near $y_0$, we obtain
\[
\d_{n_T}(\a) = \frac {2\a}{\psi_0'(\tau_{n_T})} + \Oc(\tau_{n_T}\inv) = \a\Phi(y_0) + \Oc(r_T) + \Oc(\tau_{n_T}\inv) = \a \sigma + \Oc(\tau_{n_T}^{-\frac 1 {2N} + 2\e}).
\]
This concludes the proof.
\end{proof}

\end{document}